\documentclass[journal]{IEEEtran}
\usepackage{graphicx}
\usepackage{amssymb,amsmath}
\usepackage{epstopdf}
\usepackage{subfig}  
\usepackage{tikz}
\usetikzlibrary{positioning}

\usepackage{pgfplots}
\pgfplotsset{compat=1.5}
\usetikzlibrary{pgfplots.groupplots,shapes,positioning,decorations.pathmorphing,automata}
\usepackage{pgf}
\usepackage{url}
\usetikzlibrary{arrows,automata,plotmarks}

\usepackage{amsthm}
\newtheorem{theorem}{Theorem}
\newtheorem{lemma}{Lemma}

\usepackage[nolist,nohyperlinks]{acronym}

\newcommand{\edit}[1]{{\color{black}#1}}

\title{\edit{Joint Scheduling and Coding For Low In-Order Delivery Delay Over Lossy Paths With Delayed Feedback}}

\author{Pablo Garrido$^\dag$, Douglas J. Leith$^\$$, Ram\'on Ag\"uero$^\dag$\\[6pt]
	\begin{minipage}{0.5\textwidth}
		\begin{center}
			$^\dag$ Dept. of Communications Engineering \\
			University of Cantabria, Santander 39005, Spain
			\texttt{\small \{pgarrido,ramon\}@tlmat.unican.es}
		\end{center}
	\end{minipage}
	\begin{minipage}{0.5\textwidth}
		\begin{center}
			$^\$$ School of Computer Science and Statistics\\
			Trinity College Dublin, Dublin 2, Ireland \\
			\texttt{\small doug.leith@tcd.ie}
		\end{center}
	\end{minipage}
}

\usetikzlibrary{fit}
\tikzset{%
	highlight/.style={rectangle,rounded corners,fill=red!15,draw,fill opacity=0.5,thick,inner sep=0pt}
}

\begin{document}
	\maketitle
	
	\begin{abstract}

We consider the transmission of packets across a lossy end-to-end network path so as to achieve low in-order delivery delay.    This can be formulated as a decision problem, namely deciding whether the next packet to send should be an information packet or a coded packet.  Importantly, this decision is made based on \emph{delayed} feedback from the receiver. While an exact solution to this decision problem is challenging, we exploit ideas from queueing theory to derive scheduling policies based on prediction of a receiver queue length that, while suboptimal, can be efficiently implemented and offer substantially better performance than state of the art approaches. We obtain a number of useful analytic bounds that help characterise design trade-offs and our analysis highlights that the use of prediction plays a key role in achieving good performance in the presence of significant feedback delay.  Our approach readily generalises to networks of paths and we illustrate this by application to multipath transport scheduler design.

\end{abstract}


\section{Introduction}

In this paper we revisit the transmission of packets across a lossy end-to-end network path so as to achieve low in-order delivery delay.   Consideration of end-to-end packet transmission is motivated by improving operation at the transport layer and with this in mind we also assume the availability of feedback from client to server.   This feedback is delayed by the path propagation delay and, in contrast to the link layer, this feedback delay may be substantial.  For example, on a 50Mbps path with 25ms RTT there are around 100 packets in flight and so the server only learns of the fate of a packet after a further 100 packets have been sent. In other words, the server has to make \emph{predictive} decisions about what to transmit in those 100 packets, in particular whether they are information or redundant/coded packets.
Information theory tells us that we do not need to make use of feedback in order to be capacity achieving in a packet erasure channel.  However, it also tells us that feedback can be used to reduce in-order delivery delay, possibly very considerably~\cite{Karzan2017}.  More generally, there is a trade-off between rate and delay, and feedback can be used to modify this trade-off, and it is this which is of interest.

While much attention in 5G has been focused on the physical and link layers, it is increasingly being realised that a wider redesign of network protocols is also needed in order to meet 5G requirements.   Transport protocols are of particular relevance for end-to-end performance, including end-to-end latency.   For example, ETSI have recently set up a working group to study next generation protocols for 5G \cite{ETSI2016}.  The requirement for major upgrades to current transport protocols is also reflected in initiatives such as Google QUIC \cite{Langley2017} and the Open Fast Path Alliance \cite{OpenFast2016} as well as by recent work such as \cite{kim14}.   In part, this reflects the fact that low delay is already coming to the fore in network services.  For example, Amazon estimates that a 100ms increase in delay reduces its revenue by 1\% \cite{Amazon2013}, Google measured a 0.74\% drop in web searches when delay was artificially increased by 400ms \cite{Google2011} while Bing saw a 1.2\% reduction in per-user revenue when the service delay was increased by 500ms \cite{Bing2009}.   But the requirement for low latency also reflects the needs of next generation applications, such as augmented reality and the tactile Internet.

As we will describe in more detail shortly, by use of modern low-delay streaming code constructions, the task at the transport layer can be formulated as one of deciding whether the next packet to send should be an information packet or a coded packet, with this decision being made based on stale/delayed feedback from the receiver.    The use of feedback in ARQ has of course been well studied, but primarily in the case of instantaneous feedback i.e. where there is no delay in the server receiving the feedback.   When feedback is delayed the problem becomes significantly more challenging, and has received almost no attention in the literature (notable exceptions include \cite{Vasudevan2010,Sahai2008,Leith2016}).   While the decision task can be formulated as a dynamic programming problem, the complexity grows combinatorially with the delay\footnote{\edit{In the presence of feedback delay $d$ the state space of the dynamic programme corresponds to the possible outcomes of the $d$ packets in flight (for which no feedback is yet available), the number of which grows combinatorially with $d$.}} and so quickly becomes unmanageable for even quite small delays.  In particular, such solutions are unsuited to the real-time decision-making required within next generation networks.   

\edit{In this paper we take a different approach and make use of a helpful connection between coding and queuing theory.  We use this connection to derive scheduling policies based on the prediction of the receiver queue length that, while suboptimal, can be efficiently implemented and offer substantially better performance than state of the art solutions. }  This approach also allows us to obtain a number of useful analytic bounds that help characterise design trade-offs.   Our analysis highlights that the use of prediction plays a key role in achieving good performance in the presence of significant feedback delay, and that it is prediction errors that drive the rate-delay trade-off. \edit{To the best of our knowledge this work is the first to make use of prediction with delayed feedback.} Although our main focus is on single paths, our approach readily generalises to networks of paths and we illustrate this by application to multipath transport scheduler design.

	\section{Related Work}
\label{sec:related_work}

The literature contains several different proposals for coding schemes that make use of feedback. For instance, Sundarajan \emph{et al.} introduce in~\cite{Sundararajan2017} a new linear coding scheme that includes feedback. They exploit it so that the encoder learns the packets that have been ``seen'' by the receivers, thus speeding the decoding process. \edit{A similar approach, considering wireless multicast communications, is described in \cite{Wu2015}, which proposes a joint coding/feedback scheme, scalable with respect to the number of receivers.} The authors of~\cite{Sorensen2012} propose an extension of LT and Raptor Codes that adds information feedback, with the objective of reducing the coding overhead. Hagedorn \emph{et al.} present in~\cite{Hagedorn2009} a generalized LT coding scheme that relies on feedback information. Other interesting approaches include Hybrid ARQ~\cite{Rowitch2000}, which combines a forward error correction scheme with automatic repeat-request.  A recent work that promotes the use of Hybrid ARQ for low latency and ultra reliable applications is, for example, that from Cabrera~\emph{et al}~\cite{Cabrera2017}. 

However, most of the existing literature does not consider the impact of feedback delay.  Under circumstance with no delayed feedback, it is well known that ARQ is optimal both in terms of capacity and delay~\cite{Vasudevan2010}. However, when feedback is delayed the situation changes fundamentally, and the end-to-end delay with ARQ can greatly increase. The use of coding schemes can reduce this end-to-end delay, even when the feedback delay is not small~\cite{Vasudevan2010}. The importance of considering feedback is also considered by~\cite{Sahai2008}, where the authors studied how the performance of block-coding varies with and without feedback, especially when considering the impact of delayed feedback. 
 
The analysis of coding schemes with delayed feedback remains largely open. In~\cite{Leith2016} the authors study the throughput and end-to-end delay of a variable-length block coding scheme, focusing on regimes where the feedback delay was shorter than the minimum block size. In addition, the authors focus on saturated network conditions, where the sender has an unlimited number of packets waiting to be sent.

	\section{Preliminaries}



\subsection{Low Delay Streaming Codes}
\label{sec: RandomLinearConstruction}

We model an end-to-end network path as a packet erasure channel (packets carry a unique sequence number and a checksum thus losses can be detected). Most previous works on packet erasure channels have been based on use of block codes, whereby the sequence of information packets to be transmitted is partitioned into blocks of size $k$ and $n-k$ coded packets are appended to these to create a block of size $n$ information plus coded packets, which implies a code with rate $k/n$, see Fig. \ref{Fig:introa}.   
As already noted, the requirement for low latency in next generation networks has led to renewed interest in whether alternative code constructions can yield a more favourable trade-off between throughput and in-order delivery delay.  To see that this may indeed be the case let us consider, for example, a rate $\frac{k}{n}$ systematic block code and suppose that the code is an ideal one in the sense that receipt of any $k$ of the $n$ packets allows all of the $k$ information packets to be reconstructed.  Furthermore, assume that the first information packet is lost. All remaining information packets have to be buffered until the first coded packet is received. At this point, the first information packet can be reconstructed and all of the information packets can be delivered in-order.   The in-order delivery delay is therefore proportional to $k$.  Alternatively, suppose that the $n-k$ coded packets are distributed uniformly among the information packets, rather than all being placed after the $k$ information packets, see Fig. \ref{Fig:introb}.  To keep the code causal, suppose that each coded packet only protects the preceding information packets in the block\footnote{Thus, coded packet $c_1$ protects information packets $u_1$ and $u_2$, coded packet $c_2$ protects $u_1$, $u_2$, $u_3$ and $u_4$, and so on.  Note that the resulting code construction is \emph{not} the same as using a short classical block code with $k=2$ and $n=3$ as then $c_2$ would only protect $u_3$ and $u_4$.}.  Assume again that the first information packet is lost. This loss can now be recovered on receipt of the first coded packet resulting in a delay that is now proportional to $\frac{k}{n-k}$ (i.e, this is much lower than $k$ when $n$ is large).  

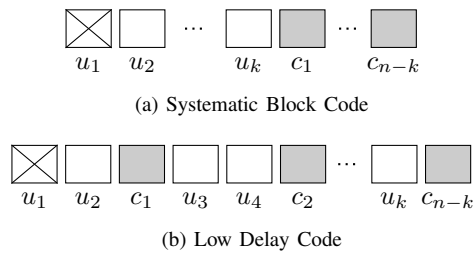
\begin{figure}
	\centering
	\subfloat[Systematic Block Code]{\label{Fig:introa}
		\begin{tikzpicture}

\node[draw,fill=white,minimum height=0.5cm,minimum width=0.6cm, align=center] (u1) at (0,0)  {};
\node[below = 0cm of u1] {$u_1$};
\draw (u1.north east) -- (u1.south west);
\draw (u1.north west) -- (u1.south east);
\node[draw,right = 0.1cm of u1,fill=white,minimum height=0.5cm,minimum width=0.6cm, align=center] (u2) {};
\node[below = 0cm of u2] {$u_2$};
\node[draw,right = 0.8cm of u2,fill=white,minimum height=0.5cm,minimum width=0.6cm, align=center] (uk) {};
\node[below = 0cm of uk] {$u_k$};
\node[draw,right = 0.1cm of uk,fill=black!20,minimum height=0.5cm,minimum width=0.6cm, align=center] (c1) {};
\node[below = 0cm of c1] {$c_1$};
\node[draw,right = 0.6cm of c1,fill=black!20,minimum height=0.5cm,minimum width=0.6cm, align=center] (cnk) {};
\node[below = 0cm of cnk] {$c_{n-k}$};

\draw[dotted,thick] ([xshift=0.26cm]u2.0) -- ([xshift=-0.26cm]uk.180);
\draw[dotted,thick] ([xshift=0.18cm]c1.0) -- ([xshift=-0.16cm]cnk.180);

\end{tikzpicture}
	}\\
	\subfloat[Low Delay Code]{\label{Fig:introb}
		\begin{tikzpicture}

\node[draw,fill=white,minimum height=0.5cm,minimum width=0.6cm, align=center] (u1) at (0,0)  {};
\node[below = 0cm of u1] {$u_1$};
\draw (u1.north east) -- (u1.south west);
\draw (u1.north west) -- (u1.south east);
\node[draw,right = 0.1cm of u1,fill=white,minimum height=0.5cm,minimum width=0.6cm, align=center] (u2) {};
\node[below = 0cm of u2] {$u_2$};
\node[draw,right = 0.1cm of u2,fill=black!20,minimum height=0.5cm,minimum width=0.6cm, align=center] (c1) {};
\node[below = 0cm of c1] {$c_1$};
\node[draw,right = 0.1cm of c1,fill=white,minimum height=0.5cm,minimum width=0.6cm, align=center] (u3) {};
\node[below = 0cm of u3] {$u_3$};
\node[draw,right = 0.1cm of u3,fill=white,minimum height=0.5cm,minimum width=0.6cm, align=center] (u4) {};
\node[below = 0cm of u4] {$u_4$};
\node[draw,right = 0.1cm of u4,fill=black!20,minimum height=0.5cm,minimum width=0.6cm, align=center] (c2) {};
\node[below = 0cm of c2] {$c_2$};

\node[draw,right = 0.6cm of c2,fill=white,minimum height=0.5cm,minimum width=0.6cm, align=center] (uk) {};
\node[below = 0cm of uk] {$u_k$};
\node[draw,right = 0.1cm of uk,fill=black!20,minimum height=0.5cm,minimum width=0.6cm, align=center] (cnk) {};
\node[below = 0cm of cnk] {$c_{n-k}$};

\draw[dotted,thick] ([xshift=0.16cm]c2.0) -- ([xshift=-0.16cm]uk.180);

\end{tikzpicture}
	}
	\caption{ Example of two codes with different throughput-delay characteristics. Shaded squares indicated coded packets, unshaded indicate information packets. }
	\vspace{-15pt}
	\label{fig:intro}
\end{figure}

With the aim of obtaining an improved trade-off between rate and delay, \cite{Karzan2017} recently proposed an alternative code construction for packet erasure channels, referred to as a streaming code (a form of convolutional code). The code is constructed by interleaving information packets $u_j$, $j=1,2,\ldots$ with coded packets $c_i$, $i=1,2,\ldots$.  One coded packet is inserted after every $l-1$ information packets and transmitted over the network path, resulting in a code of rate $\frac{l-1}{l}$. Fig. \ref{fig:scheme} illustrates this code construction. 
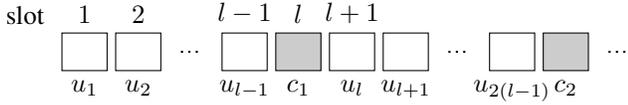
\begin{figure}
	\centering
	\begin{tikzpicture}

\node[draw,fill=white,minimum height=0.5cm,minimum width=0.6cm, align=center] (u1) at (0,0)  {};
\node[below = 0cm of u1] {$u_1$};
\node[above = 0cm of u1] (f_s) {$1$};
\node[left = 0.2cm of f_s] {slot};
\node[draw,right = 0.1cm of u1,fill=white,minimum height=0.5cm,minimum width=0.6cm, align=center] (u2) {};
\node[below = 0cm of u2] {$u_2$};
\node[above = 0cm of u2] {$2$};
\node[draw,right = 0.8cm of u2,fill=white,minimum height=0.5cm,minimum width=0.6cm, align=center] (ulm1) {};
\node[below = 0cm of ulm1] {$u_{l-1}$};
\node[above = 0cm of ulm1] {$l-1$};
\node[draw,right = 0.1cm of ulm1,fill=black!20,minimum height=0.5cm,minimum width=0.6cm, align=center] (c1) {};
\node[below = 0cm of c1] {$c_1$};
\node[above = 0cm of c1] {$l$};
\node[draw,right = 0.1cm of c1,fill=white,minimum height=0.5cm,minimum width=0.6cm, align=center] (ul) {};
\node[below = 0cm of ul] {$u_l$};
\node[above = 0cm of ul] {$l+1$};
\node[draw,right = 0.1cm of ul,fill=white,minimum height=0.5cm,minimum width=0.6cm, align=center] (ulp1) {};
\node[below = 0cm of ulp1] {$u_{l+1}$};
\node[draw,right = 0.8cm of ulp1,fill=white,minimum height=0.5cm,minimum width=0.6cm, align=center] (u2lm1) {};
\node[below = 0cm of u2lm1] {$u_{2\left(l-1\right)}$};
\node[draw,right = 0.1cm of u2lm1,fill=black!20,minimum height=0.5cm,minimum width=0.6cm, align=center] (c2) {};
\node[below = 0cm of c2] {$c_2$};

\draw[dotted,thick] ([xshift=0.26cm]u2.0) -- ([xshift=-0.26cm]ulm1.180);
\draw[dotted,thick] ([xshift=0.26cm]ulp1.0) -- ([xshift=-0.26cm]u2lm1.180);
\draw[dotted,thick] ([xshift=0.26cm]c2.0) -- ([xshift=0.5cm]c2.0);

\end{tikzpicture}
	\caption{Illustrating the low delay streaming code setup.  Sequence $\{u_j\}$ of information packets is interleaved with sequence $\{c_i\}$ of coded packets (indicated as shaded) and transmitted.  Slots correspond to a single packet transmission and are indexed $1,2,\ldots$. }
	\vspace{-10pt}
	\label{fig:scheme}
\end{figure}
 Coded packet $c_i$ can only recover an erasure of packets already transmitted and it is generated by taking random linear combinations of the previously transmitted information packets within the coding window $\{u_{L},\ldots,u_{(l-1)i}\}$, where $L$ represents the first packet protected, the coding window could be reduced by setting $L$ as the last packet acknowledged by the receiver. With the left-hand edge of the coding windows equals to 1 ($L=1$) a coded packet is generated by:

\begin{equation}
c_i = f_{ i}(u_1, u_2 , \dots ,u_{(l-1)i } ) := \sum_{j=1}^{(l-1)i} w_{ij} u_j\label{eq:code}
\end{equation}

\noindent where each information packet $u_j$ is treated as a vector in $\mathbb{F}_Q$ and each coefficient $w_{ij}\in \mathbb{F}_Q$ is chosen randomly from an i.i.d. uniform distribution, with $\mathbb{F}$ an appropriate choice of finite field, for instance $GF(2^8)$. 

Note that in practice the left-hand edge $L$ of the coding window can made be larger than $1$.   In particular, suppose that the receiver has received or decoded all information packets up to and including packet $u_j$.  Feedback can be used to communicate this to the transmitter allowing it to use $L={j+1}$ for all subsequent coded packets.   The generator matrix shown in Fig. \ref{fig:sliding-window} illustrates this sliding window approach, where the columns indicate the information packets that need to be sent and the rows indicate the composition of the packet transmitted at any given time.  

%

\begin{figure}
	\centering
	\begin{tikzpicture}[scale=0.9, transform shape]
\node(1){$
\left[\begin{array}{*8{c}}
1 &0 &0 &0 &0 &0 & 0 &0 \\
0 &1 &0 &0 &0 &0 & 0 &0 \\
0 &0 &1 &0 &0 &0 & 0 &0 \\
0 &0 &0 &1 &0 &0 & 0 &0 \\
w_{1,1} & w_{1,2} & w_{1,3} & w_{1,4} &0 &0 &0 &0\\
0 &0 &0 &0 &1 &0 &0 &0 \\
0 &0 &0 &0 &0 &1 &0 &0 \\
0 &0 &0 &0 &0 &0 &1 &0 \\
0 &0 &0 &0 &0 &0 &0 &1 \\
0 &0 &w_{2,3} & w_{2,4} & w_{2,5} & w_{2,6} & w_{2,7} & w_{2,8}
\end{array}\right]$};

\node(2)[above = 0cm of 1]{$\color{white}\left[\begin{array}{*8{c}}
	w_{1,1} & w_{1,2} & w_{1,3} & w_{1,4} &0 &0 &0 &0\\
	0 &0 &w_{2,3} & w_{2,4} & w_{2,5} & w_{2,6} & w_{2,7} & w_{2,8}\\
	\color{black} u_1 & \color{black}u_2 & \color{black}u_3 & \color{black}u_4 & \color{black}u_5 &\color{black} u_6 & \color{black}u_7 & \color{black}u_8
\end{array} \color{white}\right]$};
\node(3)[above = 0.5cm of 1]{Information Packets};
\node(4)[left = -0.5cm of 1]{$\begin{array}{*1{c}}
1\\
2\\
3\\
4\\
5\\
6\\
7\\
8\\
9\\
10
\end{array}$};
\node(5)[left = 0.5cm of 1, rotate=90]{Time};
\end{tikzpicture}
	\caption{Example generator matrix for the low delay code with sliding window showing the coefficients used to produce each packet. In this example, we assume that the transmitter has obtained knowledge from the receiver by time 10 indicating that it has successfully received/decoded packets $u_1$ and $u_2$ allowing it to adjust the left-hand edge of the coding window to exclude them from packet $c_2$. \edit{Image adapted from \cite{Karzan2017}.}}
	\vspace{-10pt}
	\label{fig:sliding-window}
\end{figure}

The receiver decodes on-the-fly once enough packets/degrees of freedom have been received. In more detail, the receiver maintains a generator matrix $G_t$ at time $t$, which is similar to that shown in Fig. \ref{fig:sliding-window} except that it is composed only of the coefficients obtained from received packets. If $G_t$ is full rank, Gaussian elimination is used to recover from any packet erasures that may have occurred during transit.   We will make the standing assumption that the field size $Q$ is sufficiently large that with probability approaching one each coded packet helps the receiver recover from one information packet erasure i.e. each coded packet row added to generator matrix $G_t$ increases the rank of $G_t$ by one. 

In summary, this streaming code construction generates coded packets that are (i) individually streamed between information packets (rather than being transmitted in groups of size $n-k$ packets) and (ii) each coded packet protects all preceding information packets (rather than just the information packets within its block).  See \cite{Karzan2017} for a detailed analysis of the throughput and delay performance of this code, but for a given code rate it is easy to see that this code construction tends to decrease the overall in-order delivery delay at the receiver compared to a block code, as illustrated in the example above.    

\subsection{Decision Problem}\label{sec:decision}

Our interest in the above streaming code construction is twofold.  Firstly, for a given coding rate under a wide range of conditions it offers lower in-order deliver delay compared to standard block codes \cite{Karzan2017}.  Thus it provides a useful starting point for developing methods for low delay transmission across lossy network paths.    Secondly, it lends itself to being embedded within a clean decision problem.  Namely, one where rather than transmitting coded packets periodically according to a predetermined schedule, at each transmission opportunity the transmitter dynamically decides whether to send an information packet or a coded packet based on feedback from the receiver\footnote{Use of block codes leads to a significantly more complex decision problem.  To see this observe that losing more than \edit{$n-k$} packets within a block requires transmission of additional coded packets from that block in order to avoid a decoding failure.  These are then received interleaved with later blocks.  Thus we lose the renewal structure of open-loop block code constructions and the decision-maker needs to (i) keep track of multiple generations of interleaved blocks, each perhaps of a different size, and (ii) decide from which block to send a coded packet as well as deciding whether to send an information or coded packet.}.

Formally, assume a time-slotted system where each slot corresponds to transmission of a packet.  We have an arrival process consisting of a sequence of information packets $\{A_k,k=1,2,\dots\}$, where $A_k\in\{0,1\}$ is the number of new information packets in slot $k$, and define $\bar{a}:=\lim_{k\rightarrow\infty}\frac{1}{k}\sum_{i=1}^{k}A_i$ as the average arrival rate. These information packets are buffered at the transmitter and then sent across a lossy path to a receiver. The queue occupancy $Q^t_k$ at the transmitter\footnote{\edit{Note that packets dequeued from $Q^t$ are held in an encoding buffer at the transmitter until the receiver has signalled that they have been successfully received and so the left-hand edge $L$ in (\ref{eq:code}) can be updated, see earlier discussion. }} in slot $k$ behaves according to:
\begin{align}
Q^t_{k+1} = [Q^t_k + A_k - S_k]^+
\label{eq:Qt}
\end{align}
where $S_k\in\{0,1\}$ is the number of information packets transmitted in slot $k$ and $Q^t_1=0$. We let $\bar{s}:=\lim_{k\rightarrow\infty}\frac{1}{k}\sum_{i=1}^{k}S_i$ denote the average transmit rate.

Define a random variable $X_k$, which takes value $1$ when a packet transmitted in slot $k$ is erased and $0$ otherwise. We will assume the sequence of random variables $\{X_k\}$ is i.i.d. $X_k \sim X$ with $\text{Prob} \left(X=1\right)=p$, and that when $p=0$ then $X_k=0$ for all $k$ (so as $k\rightarrow\infty$ the occurrence of a non-zero but finite number of losses is excluded).   

Received packets are buffered at the receiver until they can be delivered in-order to an application i.e. when an information packet is erased then subsequently arriving information packets are buffered until the lost packet can be recovered.   A coded packet sent in slot $k$ is built as the random linear combination of all information packets sent before slot $k$. In each slot $k$ the receiver also sends feedback to the transmitter, informing of the packets already received as of slot $k$. This feedback arrives at the transmitter after delay $d$, in slot $k+d$. It is assumed, for simplicity, that none of these feedback packets are lost.

Fig.~\ref{Fig: Schematic} illustrates this problem setup.  In each slot $k$ the transmitter has the choice of (i) doing nothing, (ii) sending the information packet at the head of the transmitter queue, or (iii) sending a coded packet.     Our task is to solve the transmitter decision problem while satisfying a number of constraints: both the transmitter and receiver queues are stabilized, the link capacity is respected, and the buffering delay at the receiver is kept small.

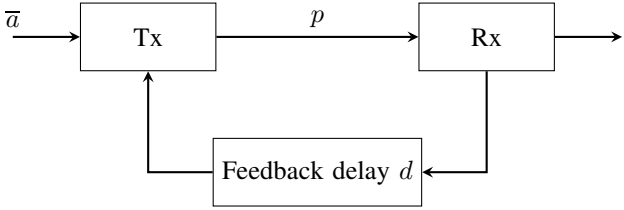
\begin{figure}
	\newlength{\mywidth}	
	\setlength{\mywidth}{0.9cm}
	\begin{tikzpicture}
	
	\tikzstyle{box node} = [rectangle, draw=black, minimum height=\mywidth, minimum width=2\mywidth]
	\node [box node] (tx) at (0,0) {Tx};
	\node [box node] (rx) at +(5\mywidth,0) {Rx};
	\node [box node] (fd) at (2.5\mywidth, -2\mywidth) {Feedback delay $d$};
	
	\draw[thick, ->, >=stealth] (-2\mywidth,0) node[above]{$\overline{a}$}  -- (tx.west) ;			
	\draw[thick, ->, >=stealth] (rx) --  (7\mywidth,0);
	\draw[thick, ->, >=stealth] (tx) -- (rx) node[pos=0.5,above]{$p$};
	\draw[thick, ->, >=stealth] (rx.south) -- (5\mywidth, -2\mywidth)  -- (fd.east);	
	\draw[thick, ->, >=stealth] (fd.west) -- (0\mywidth, -2\mywidth) -- (tx.south);	
	\end{tikzpicture}
	\caption{Schematic of the decision problem setup.  Packets arrive at Tx with mean rate $\bar{a}$, are transmitted from Tx to Rx and may be erased with probability $p$.  Rx informs Tx of its state via feedback, which is delayed by $d$ slots.}
	\label{Fig: Schematic}
\end{figure}

\section{Low Delay Scheduling Policies}

\subsection{Introduction}

When the feedback delay is zero then the decision problem in Fig. \ref{Fig: Schematic} is akin to ARQ, which of course has been well studied and for which fairly complete results are known.  However, situations where the feedback delay is non-zero have received far less attention in the literature.  In part this is because most work has focussed on the link layer where feedback delays are low, plus it is well known that open-loop block codes (which do not use feedback) are capacity achieving.  And in part this is because of the complexity of the decision problem with delayed feedback, which grows combinatorially with the feedback delay.   As already noted, next generation transport protocols seek to achieve low delay transmission over end-to-end paths.   This means that they are required to operate with significant delays before feedback is received.  This, together with our observation in Section \ref{sec:decision} that the low delay streaming code construction in Section \ref{sec: RandomLinearConstruction} lends itself to the use of feedback to make more refined decisions as to when to send coded packets, motivates revisiting the analysis and design of schedulers using delayed feedback.   

A basic difficulty is that the complexity of deciding on an optimal packet schedule grows exponentially with the feedback delay.   This means that optimal decision-making quickly becomes unmanageable for real-time operation.   Ad hoc heuristic approaches are of course possible, but they typically remain difficult to analyze and come with few performance guarantees. To make progress we make use of the observation that the decoding process at the receiver can be modelled using a queueing approach.  Namely, information packets arriving at the receiver are delivered in-order to an application until an information packet is lost, at which point subsequent information packets are buffered until the lost packet can be recovered.   Each arriving coded packet can repair the loss of any one preceding information packet, with decoding taking place once the number of received coded packets matches the number of erased information packets.   \edit{We thus define a virtual queue at the receiver, with occupancy $Q^r_k$, which behaves according to:
\begin{align}
Q^r_{k+1} = [Q^r_k + {S}_k \cdot X_k - C_k(1-X_k)]^+
\end{align}
where $X_k=1$ when packet $k$ is erased (lost) and $0$ otherwise, $S_k=1$ when an information packet is sent in slot $k$, $C_k=1$ if a coded packet is sent, while $C_k=S_k=0$ when no transmission is made. The queue occupancy $Q^r_k$ increases whenever an information packet is deleted and decreases if a coded packet is successfully received. Decoding events occur at slots $k$ where $Q^r_k=0$.  While low queue occupancy is, by itself, no guarantee of low decoding delay, in practice it tends to encourage frequent emptying of the virtual queue and so short decoding delay.}

Intuitively, the length of this virtual queue is correlated with the in-order delivery delay at the receiver -- as $Q^r_k$ grows the number of information packets buffered at the receiver will also tend to grow.   The relationship is not one to one, and we explore it further in the next section, but as we will see it is sufficient to form the basis of simple yet effective scheduling policies.   Importantly, by taking this approach we are able to obtain bounds on delay and rate which can be used for analysis and design.

\subsection{Relating Delay and Queue Occupancy}
\label{Sec: relating delay}

We proceed by considering in more detail the relationship between end-to-end in-order delivery delay, the transmitter queue occupancy $Q^t_k$ and the receiver virtual queue occupancy $Q^r_k$.  First, observe that the end-to-end delay can be divided into: (i) the time between being enqueued at the sender and being first transmitted, $D_{qt}$, and (ii) the time between being first transmitted and when the packet is successfully delivered to the application layer, $D_{qr}$.  We expect that $D_{qt}$ is related to $Q^t_k$ and $D_{qr}$ with $Q^r_k$, and indeed this can be seen in Fig.~\ref{Fig: ComparisonMuLambdaDelay}. This figure plots the average of the delays, after repeating the experiment 100 times, $D_{qt}$ and $D_{qr}$ per packet Vs. the queue occupancies $Q^t_k$ and $Q^r_k$ over a path with erasure rate $p=0.2$ and with coded packets sent periodically every $p/(1-p)$ information packets. Also indicated is the 95\% confidence interval.
%
 %
The strong correlation between delay and queue occupancy is clearly evident.  Further, it can be seen that the impact of the receiver queue occupancy $D_{qr}$ on delay is much larger than that of the transmitter queue $D_{qt}$.   This is perhaps to be expected, since a loss causes all subsequent information packets to be delayed at the receiver until the loss is repaired and decoding takes place ($Q^r_k$ becomes zero), hence amplifying the effect of a non-zero queue occupancy $Q^r_k$ on delay.    Although the data in Fig.~\ref{Fig: ComparisonMuLambdaDelay} is for a particular choice of loss and arrival rate it is representative of the behaviour seen for other choices.

\begin{figure}
	\def \figurewidth {0.75\columnwidth}
	\def \figureheight {3.0cm}
	\pgfplotstableread{figures/data/delay_lambda_p02_a07_d0.dat}{\datasetA} 
\pgfplotstableread{figures/data/delay_mu_p02_a07_d0.dat}{\datasetB} 
\begin{tikzpicture} 
font = \scriptsize,

\begin{axis}[scale only axis, 
width=\figurewidth,
height=\figureheight,
mark options={solid},
ymin=0,
ymax=150,
xmin=0,
xmax=9,
ylabel={e2e delay (\# slots)},
xlabel={Queue size, ($Q^t$ ,  $Q^r$) (\# pkts)},
compat=1.3,
legend style={legend pos = north west, draw=none, fill=none , legend columns=2, font = \scriptsize}
]



\addplot [color=black!30, line width=1.5pt]
plot [error bars/.cd, y dir = both, y explicit]
table[x index = 0, y index =1, y error index=2, y error minus index=2] from \datasetA;
\addlegendentry{$D_{qt}$}

\addplot [color=black, line width=1.5pt]
plot [error bars/.cd, y dir = both, y explicit]
table[x index = 0, y index =1, y error index=2, y error minus index=2] from \datasetB;
\addlegendentry{$D_{qr}$}

\end{axis}

\end{tikzpicture}
	\caption{Impact of queue lengths $Q^t_k$ and $Q^r_k$ on the average delay at the transmitter, $D_{qt}$, and the receiver, $D_{qr}$. In this experiment, erasure rate is $p=0.2$, arrival rate is $\overline{a} = 0.7$.   }
	\label{Fig: ComparisonMuLambdaDelay}
\end{figure}
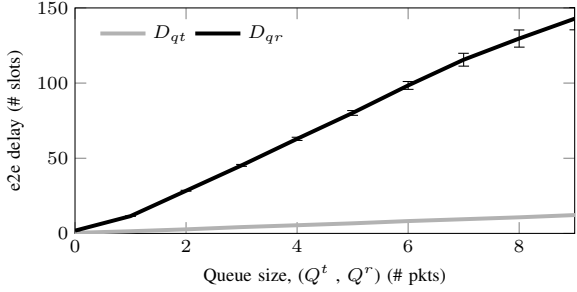

\subsection{Transmission Policies}
\label{subsec:transmissionpolicies}

Based on the insight provided by the above analysis we consider the following class of transmission policies:

\begin{align}
&C_k \in \arg\min_{C\in\{0,1\}}  F({Q}^r_{k-d},\hat{Q}^r_k,Q^t_k)C \\
&\hat{Q}^r_{k}= \hat{\theta}({Q}^r_{k-d})\\
&S_k=\min\{Q_k^t+A_k,1-{C}_k\}
\end{align}
where function $F(\cdot)$ is a design parameter, which we will discuss in more detail shortly. Observe that selection of $C_k$ uses only information available at the sender at time $k$.   Since $S_k=\min\{Q_k^t+A_k,1-{C}_k\}$, an information packet is transmitted when (i) $1-C_k=1$, and (ii) the transmission queue contains a packet to be sent. Furthermore, $Q^r_{k-d}$ is only available at the sender after feedback delay $d$. We will focus on the estimator

\begin{align}
\hat{Q}^r_k = \hat{\theta}(Q^r_{k-d})={Q}^r_{k-d}+\sum_{j=k-d}^{k-1}({S}_{j}p-C_j(1-p))\label{eq:est}
\end{align}
which simplifies to $\hat{\theta}(Q^r_{k-d})=Q^r_k$ when the feedback delay $d=0$.   This estimator makes a $d$-step ahead prediction of the value of $Q^r_k$ based on ${Q}^r_{k-d}$ and the average path loss $p$.    We will consider the impact of the accuracy of estimator predictions in more detail shortly.  Other choices of estimator are of course possible, but (\ref{eq:est}) has the virtues of simplicity and tractability.   

This class of transmission policies includes ARQ and open-loop FEC as special cases. Namely, when $F({Q}^r_{k-d},\hat{Q}^r_k,Q^t_k)=-\hat{Q}^r_k$ and $d=0$, then $C_k=1$ when ${Q}^r_k>0$ i.e. a coded packet is sent whenever the receiver reassembly queue is non-empty. Since for code construction considered this coded packet will actually be an information packet, we have ARQ. Similarly, selecting $F({Q}^r_{k-d},\hat{Q}^r_k,Q^t_k)=-(\hat{Q}^r_{k}-{Q}^r_{k-d})$ then as $d\rightarrow\infty$ we recover the open-loop FEC in~\cite{Karzan2017}, whereby a coded packet is sent every $p/(1-p)$ information packets. To see this, observe that $C_k=1$ when $\hat{Q}^r_{k}-{Q}^r_{k-d}=p(\sum_{j=k-d}^{k-1}({S}_{j}-C_j(1-p)/p)>0$.   

Recall from Section~\ref{Sec: relating delay} that the delay is much more strongly affected by the receiver queue occupancy $Q^r$ than by the transmitter queue occupancy $Q^t$.  With this in mind, Fig.~\ref{Fig: Qr_Qt_delay} compares the end-to-end system delay for different transmission policies. First, we take ARQ as a baseline scheme, comparing it with $F(Q^r_{k-d}, \hat{Q}^r_k, Q^t_k) = \rho \cdot Q^t_{k} - \hat{Q}^r_{k}$, where $\rho$ is a configuration parameter that modulates the weight given to the transmission queue length. As can be seen, the more weight that is given to $Q_t$ (higher $\rho$), the longer the end-to-end system delay.   This suggests that we should favour policies $P$ such that:
\begin{figure}
\centering
\subfloat[Packet delay vs. $\epsilon$ ($d=0$)]{
	\def \figurewidth {0.8\columnwidth}
	\def \figureheight {3.0cm}
	\hspace{-0.5cm}
	\pgfplotstableread{figures/Qt_roQr/ro0_p01_e001_01.dat}{\datasetZero} 
\pgfplotstableread{figures/Qt_roQr/ro001_p01_e001_01.dat}{\datasetTwentyFive}
\pgfplotstableread{figures/Qt_roQr/ro05_p01_e001_01.dat}{\datasetFifty}
\pgfplotstableread{figures/Qt_roQr/ro1_p01_e001_01.dat}{\datasetOne}

\begin{tikzpicture} 
font = \scriptsize, 
\begin{axis}[scale only axis, 
	ymode=log,
    width=\figurewidth,
    height=\figureheight,
    mark options={solid},
    ymin=0,
    minor tick num=1,
	xmin=0.01,
	xmax=0.2,
	xtick={0.01, 0.05, 0.1, 0.15, 0.2},
	xticklabels={0.01,0.05,0.1,0.15,0.2},
    ylabel={e2e delay (\# slots)},
    xlabel={$\epsilon$},
    compat=1.3,
	legend style={ legend pos = north east, draw=none, fill=none,legend columns=3}
    ]

\addplot [color=black, dotted, line width=1.5pt]
plot[]
table[x index = 0, y index =1] from \datasetZero;
\addlegendentry{\emph{ARQ}} 

\addplot [color=black!30, line width= 1.5pt]
plot[]
table[x index = 2, y index =3] from \datasetZero;
\addlegendentry{$\rho  = 0.00$}

\addplot [color=black, line width= 1.5pt]
plot[]
table[x index = 2, y index =3] from \datasetTwentyFive;
\addlegendentry{$\rho = 0.25$}



\end{axis}
\end{tikzpicture}
	}\quad\\
\subfloat[Packet delay vs. feedback delay $d$ ($\bar{a}=0.8$)]{
    \def \figurewidth {0.8\columnwidth}
	\def \figureheight {3.0cm}
    \pgfplotstableread{figures/data/delay_vs_d_a08_p01.dat}{\dataset} 

\begin{tikzpicture} 
font = \scriptsize, 
\begin{axis}[scale only axis, 
	ymode=log,
    width=\figurewidth,
    height=\figureheight,
    mark options={solid},
    ymin=0,
    ymax=100,
    minor tick num=1,
	xmin=0,
	xmax=100,
    ylabel={e2e delay (\# slots)},
    xlabel={$d$},
    compat=1.3,
	legend style={ legend pos = south east, draw=none, fill=none,legend columns=3}
    ]

\addplot [color=black, dotted, line width=1.5pt]
plot[]
table[x index = 0, y index =7] from \dataset;
\addlegendentry{\emph{ARQ}}

\addplot [color=black!30, solid, line width=1.5pt]
plot[]
table[x index = 0, y index =1] from \dataset;
\addlegendentry{$\rho  = 0.00$}

\addplot [color=black, line width= 1.5pt]
plot[]
table[x index = 0, y index =3] from \dataset;
\addlegendentry{$\rho = 0.25$}


\end{axis}
\end{tikzpicture}
}
	\caption{Comparison of packet delay vs. arrival rate $\epsilon=1-p-\bar{a}$ and feedback delay $d$ for ARQ $F(Q^r_{k-d}, \hat{Q}^r_k, Q^t_k) = -Q^r_{k-d}$ system and  $F(Q^r_{k-d}, \hat{Q}^r_k, Q^t_k) =  \rho Q^t_{k} - Q^r_{k-d}$. Loss rate $p=0.1$, ${Q}^r_{k-d} = Q^r_{k}$}
	\label{Fig: Qr_Qt_delay}
\end{figure}
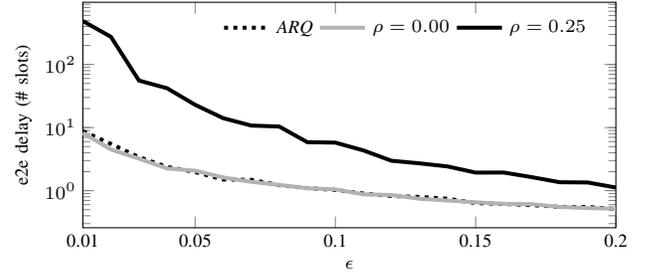
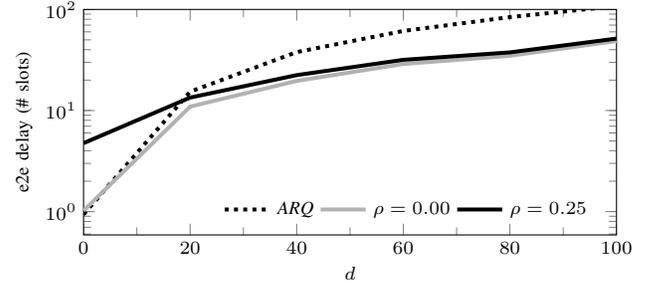

\begin{align}
&C_k \in \arg\min_{C\in\{0,1\}}  (-\hat{Q}^r_k +\gamma ) C \label{eq:up1}\\
&\hat{Q}^r_{k}=\hat{\theta}({Q}^r_{k-d})\label{eq:up2}\\
&S_k=\min\{Q_k^t+A_k,1-{C}_k\}
\end{align}
\noindent where $\gamma\ge 0$ is a design parameter. Observe that this class of policies corresponds to a threshold rule, namely ${C}_k=1$ when $\hat{Q}^r_k -\gamma  > 0$ and ${C}_k=0$ otherwise. As noted above, when $d=0$ and $\gamma= 1$ this transmission policy reduces to ARQ, while when $d\rightarrow\infty$ then it reduces to open-loop FEC. That is, in these two boundary cases this transmission policy reverts to the state of the art.

\subsection{Estimator Accuracy}

\edit{Before proceeding to analyse transmission policy $P$ we first derive some bounds on the accuracy of estimator (\ref{eq:est}) that will prove useful later.   

The following lemma is a restatement of \cite[Proposition 3.1.2]{Mey08},
\begin{lemma}[Queue Continuity]
\label{th:skorokhodcontinuity}
Consider queue updates $q_{k+1}  = [q_{k} + \omega_k]^+$ and $\tilde{q}_{k+1}   = [ \tilde{q}_k +  \tilde{\omega}_k]^+$ where $q_1  = \tilde{q}_1\ge 0$ and  $\omega_k$, $\tilde{\omega}_k\in\mathbb{R}$ are the queue increments.  Suppose $| \sum_{i=1}^k \omega_i - \tilde{\omega}_i | \le \delta/2$ for all $k$ and some $\delta \ge 0$.  Then $|q_k - \tilde{q}_k| \le \delta$, $k=1,2,\dots$.
\end{lemma}

\noindent Applying Lemma \ref{th:skorokhodcontinuity} to $Q^r_k$ and $\hat{Q}^r_k$ then $\omega_k=S_kX_k - C_k(1- X_k)$, $\tilde{\omega}_k={S}_{k}p-C_k(1-p)$ and $\sum_{i=1}^k (\omega_i-\tilde{\omega}_i)=\sum_{j=k-d}^{k-1} (S_j+C_j)(X_j-p)$.   Since $X_j\in\{0,1\}$ and $0\le S_j+C_j\le 1$ then
\begin{align}
   -dp\le \sum_{i=1}^k (\omega_i-\tilde{\omega}_i) = 
    \le d(1-p)
\end{align}
Hence, 
\begin{align}
    |Q^r_k - \hat{Q}^r_k|\le \delta = 2d\max\{p,1-p\}\label{eq:upperlimit}
\end{align}
We can obtain sharper bounds on $\sum_{j=k-d}^{k-1}(X_j-p)$ by taking more advantage of the fact that $X_j$ is a random variable.  For example, when losses are i.i.d then the $\{X_j\}$ are also i.i.d.  and we can use Hoeffding's inequality~\cite{LeonGarcia2008} applied to Bernoulli random variables to obtain 
\begin{equation}
\text{Prob}\left(\left|\sum_{j=k-d}^{k-1}(X_j-p)\right|\ge \epsilon dp\right)
\le 2e^{-2\epsilon^2 d}
\end{equation}
where $\epsilon>0$. It follows immediately that $\text{Prob}(|Q^r_k - \hat{Q}^r_k|\ge \delta)
\le 2e^{-2(\delta/2dp)^2 d}$ and so
\begin{align}
|Q^r_k - \hat{Q}^r_k |
\le \delta=2p\sqrt{\frac{d}{2}\log\left(\frac{2}{1-q}\right)}
\label{eq:hoeff}
\end{align}
with probability at least $q$.  The bound (\ref{eq:hoeff}) is generally substantially sharper than bound (\ref{eq:upperlimit}), as can be seen in Fig. \ref{fig:bounds}.
}

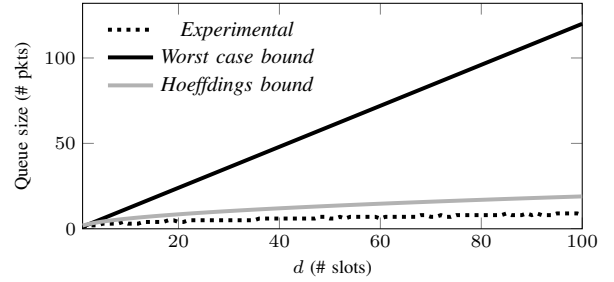
\begin{figure}
    \centering
    \def \figurewidth {0.75\columnwidth}
	\def \figureheight {3.0cm}
    \pgfplotstableread{figures/data/bound_q09_p06.dat}{\dataset} 

\begin{tikzpicture} 
font = \scriptsize, 
\begin{axis}[scale only axis, 
width=\figurewidth,
height=\figureheight,
mark options={solid},
ymin=0,
xmin=1,
xmax=100,
ylabel={Queue size (\# pkts)},
xlabel={$d$ (\# slots)},
compat=1.3,
legend style={legend pos=north west, draw=none, fill=none,legend columns=1, font = \footnotesize}
]

\addplot [color=black, dotted,line width=1.5pt]
plot[]
table[x index = 0, y index =1] from \dataset;
\addlegendentry{\emph{Experimental}}

\addplot [color=black,line width=1.5pt]
plot[]
table[x index = 0, y index =2] from \dataset;
\addlegendentry{\emph{Worst case bound}}

\addplot [color=black!30, solid, line width=1.5pt]
plot[]
table[x index = 0, y index =3] from \dataset;
\addlegendentry{\emph{Hoeffdings bound}}


\end{axis}
\end{tikzpicture}
    \caption{\edit{Comparing the bounds on $|Q^r_k - \hat{Q}^r_k|$ obtaining using worst-case analysis (\ref{eq:upperlimit}) and using Hoeffding's inequality (\ref{eq:hoeff}) (with $q=0.9$). Loss rate $p=0.6$.}}
    \label{fig:bounds}
\end{figure}

\subsection{Bounding Virtual Receiver Queue}

 \edit{Armed with these bounds on estimator accuracy we are now in a position to bound the receiver queue occupancy (recall that we have already seen that the end-to-end delay mostly depends on the receiver queue occupancy).  The following establishes that for the class of policies $P$ with estimator (\ref{eq:est}) we can upper bound the queue length by $\gamma+\delta+1$,} 
\begin{theorem}\label{lem:one}
	Consider transmission policy $P$ using estimator (\ref{eq:est}).   Suppose \edit{estimate $\hat{Q}^r_{k}$ satisfies ${Q}^r_{k}-\hat{Q}^r_{k}\le\delta$, $k=1,2,\dots$}.  When $0<p<1$ then $Q^r_{k}$ converges almost surely to the interval  $0 \le Q^r_{k}  \le  \gamma+\delta+1$ \edit{as $k\rightarrow\infty$}.
\end{theorem}
\begin{proof}
	
	Since
	$\hat{Q}^r_{k}={Q}^r_{k-d}+\sum_{j=k-d}^{k-1}({S}_{j}p-C_j(1-p))$ then 
	\begin{align}
	\hat{Q}^r_{k+1}=&\hat{Q}^r_{k}+({Q}^r_{k-d+1}-{Q}^r_{k-d}) \notag\\
	&+({S}_k-{S}_{k-d})p-(C_k-C_{k-d})(1-p)
	\end{align}
	Now ${Q}^r_{k-d+1}-{Q}^r_{k-d} = [{Q}^r_{k-d} + S_{k-d}X_{k-d} - C_{k-d}(1-X_{k-d}) ]^+-{Q}^r_{k-d}= {S}_{k-d}X_{k-d} - C_{k-d}(1-X_{k-d})$ when ${Q}^r_{k-d}\ge 1$. Hence,  
	\begin{align}
	\hat{Q}^r_{k+1} =& \hat{Q}^r_{k}+(X_{k-d}-p)(S_{k-d} + C_{k-d}) \notag\\
	&+ {S}_kp -C_k(1-p)
	\label{eq:update}
	\end{align}
	when  ${Q}^r_{k-d}\ge 1$.  Conversely, when ${Q}^r_{k-d}< 1$ then ${Q}^r_{k-d}=0$ since it is non-negative and integer valued.   Hence, ${Q}^r_{k-d+1}-{Q}^r_{k-d} = [{S}_{k-d}X_{k-d} - C_{k-d}(1-X_{k-d}) ]^+={S}_{k-d}X_{k-d}$ and 
	\begin{align}
	\hat{Q}^r_{k+1}=&  \hat{Q}^r_{k}+(X_{k-d}-p)(S_{k-d}+C_{k-d}) \notag \\
	& + S_kp-C_k(1-p) + C_{k-d}(1-X_{k-d}) 
	\label{eq:update2}
	\end{align}
	
	%
	We proceed by considering the following two cases. 
	
	Case (i):  $-\gamma +\hat Q^r_k \ge 1$.  Since $-\gamma +\hat Q^r_k >0$ then $C_k = 1, {S}_k =0$.  When ${Q}^r_{k-d}\ge 1$ then (\ref{eq:update}) applies and since $-\gamma +\hat{Q}^r_k \ge 1$ then $-\gamma +\hat{Q}^r_{k+1} = -\gamma +\hat{Q}^r_{k}+\Delta^1_k$ with $\Delta^1_k:=(X_{k-d}-p)(S_{k-d} + C_{k-d})-(1-p)\le 0$.  Similarly, when ${Q}^r_{k-d} < 1$ then $-\gamma + \hat{Q}^r_{k+1} = -\gamma + \hat{Q}^r_k + \Delta^2_k$ with $\Delta^2_k:=(X_{k-d}-p)(S_{k-d} + C_{k-d}) -(1-p) + C_{k-d}(1-X_{k-d}))\le 0$.  Therefore, 
	\begin{align}
	-\gamma +\hat Q^r_{k+1} \le -\gamma+\hat Q^r_{k}
	\end{align}
	Observe that $\Delta^1_k$ is strictly less than zero when $X_{k-d} = 0$ and $\Delta^2_k$ is strictly less than zero when $X_{k-d} = 0$ and $C_{k-d} = 0$. By assumption $0<p=\text{Prob}(X_k=1)<1$ and $X_{k-d}, \forall k$ are independent of $\hat Q^r_k$. Hence if $-\gamma +\hat Q^r_k \ge 1$ persists then, with probability one, a slot will occur where $X_{k-d}=0$ and so $\Delta^1_k<0$.  Further, when 
	$-\gamma +\hat Q^r_k \ge 1$ and $Q^r_{k-d} <  1$ then $\sum_{j=k-d}^{k-1}(p-C_j)= \sum_{j=k-d}^{k-1}(\hat{S}_{j}p-C_j(1-p))\ge \sum_{j=k-d}^{k-1}({S}_{j}p-C_j(1-p)) > 1+\gamma$.  Since $p<1$  and $\sum_{j=k-d}^{k-1}(p-C_j)>0$ it follows that $C_{j}=0$ for at least $\lceil d(1-p) \rceil$ of the slots in the sum. Therefore, regardless of $(S_{k-d} + C_{k-d})$, with positive probability over any $d$ slots a slot will occur where $X_{k-d}=0$, $C_{k-d}=0$ and $\Delta^2_k<0$.
	
	Case (ii) $-\gamma + \hat Q^r_k \le 1$.  We now have two subcases to consider:
	\begin{enumerate}
		
		\item[(a)]  When $0 \le-\gamma + \hat Q^r_k\le 1$, then $C_k=1$ and $S_k=0$. By update~(\ref{eq:update}) $\hat{Q}^r_{k+1} = \hat{Q}^r_{k} + (X_{k-d} - p)(S_{k-d} + C_{k-d}) - (1-p) \le \hat{Q}^r_{k} -1 + p$ and by update ~(\ref{eq:update2}) $\hat{Q}^r_{k+1} = \hat{Q}^r_{k} + (X_{k-d} - p)(S_{k-d} + C_{k-d}) - (1-p) + C_{k-d}(1-X_{k-d}) \le \hat{Q}^r_{k}$. Hence, $\hat{Q}^r_{k+1} \le \hat{Q}_k^r$ and, therefore, $-\gamma + \hat{Q}^r_{k+1} \le 1$
		
		\item[(b)] When $-\gamma + \hat Q^r_k \le 0$ then $S_k\in{0,1}$ and $C_k=0$. By update (\ref{eq:update}) $\hat{Q}^r_{k+1} = \hat{Q}^r_{k} + (X_{k-d} - p)(S_{k-d} + C_{k-d}) + pS_k \le \hat{Q}^r_{k} + p$ and by update (\ref{eq:update2}) $\hat{Q}^r_{k+1} = \hat{Q}^r_{k} + (X_{k-d} - p)(S_{k-d} + C_{k-d}) + pS_k + C_{k-d}(1-X_{k-d}) \le \hat{Q}^r_{k} + 1$. Therefore, $-\gamma + \hat{Q}^r_{k+1} \le 1$			
	\end{enumerate}
	
	We have that $-\gamma+\hat Q^r_{k+1}$ never increases and strictly decreases with positive probability when $ -\gamma + \hat Q^r_{k} > 1$. And when  $-\gamma+\hat Q^r_{k} \le 1$ then $ -\gamma+\hat Q^r_{k+1} $ never goes above $1$. Hence, we can conclude that $\hat Q^r_{k}$ converges almost surely and that it is indeed upper bounded by $-\gamma+\hat Q^r_{k+1} \le  1$ \edit{i.e. $\hat Q^r_{k+1} \le  \gamma+1$}.  Since \edit{${Q}^r_{k}-\hat{Q}^r_{k}\le\delta$} it follows that \edit{${Q}^r_{k}\le \hat{Q}^r_{k}+\delta \le \gamma+1+\delta$}, and the stated interval now follows from the fact that $Q^r_{k} \ge 0$.
\end{proof}

Importantly, observe that the bound in Theorem~\ref{lem:one} is in terms of the instantaneous queue length $Q^r_k$ and applies to every sample path.  It is therefore much stronger than a bound on the average queue length.   \edit{One immediate consequence of this, for example, is that the requirement that the estimator is accurate in the sense that ${Q}^r_{k}-\hat{Q}^r_{k}\le\delta$ can be relaxed to one that this only holds with a given probability $q$.   The bound in the Theorem~\ref{lem:one} then applies to those sample paths for which the estimator is sufficiently accurate \emph{i.e.} also applies with probability $q$.   For example, using this observation we can immediately use the Hoeffding's bound on estimator accuracy (\ref{eq:hoeff}) to select a value for $\delta$.  
} 

\edit{From Theorem~\ref{lem:one} it can be seen that the maximum queue length $Q_k^r$, and so delay, tends to increase with design parameter $\gamma$.  Hence, to minimise delay we should choose parameter $\gamma$ small.  This is also confirmed by simulation, e.g. see Figure \ref{Fig:gamma} which plots delay vs traffic load for various values of $\gamma$.  Bound (\ref{eq:hoeff})} tells us that the maximum queue length also tends to increase with the feedback delay $d$ and with loss rate $p$, although no more than linearly in both.

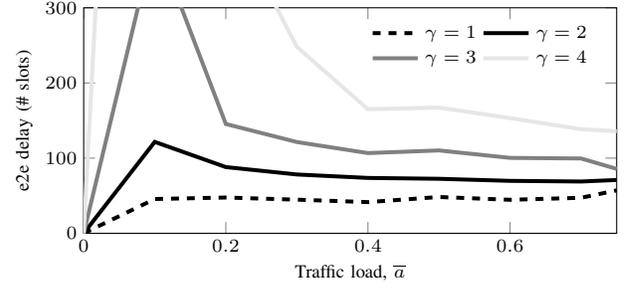
\begin{figure}
	\def \figurewidth {0.8\columnwidth}
	\def \figureheight {3.0cm}
    \pgfplotstableread{figures/data/delay_gamma_a_p02_d100.dat}{\datasetA} 
\begin{tikzpicture} 
font = \scriptsize, 

\begin{axis}[scale only axis, 
width=\figurewidth,
height=\figureheight,
mark options={solid},
ymin=0,
ymax=300,
xmin=0,
xmax=0.75,
ylabel={e2e delay (\# slots)},
xlabel={Traffic load, $\overline{a}$},
compat=1.3,
legend style={legend pos = north east, draw=none, fill=none , legend columns=2, font = \scriptsize}
]


\addplot [color=black, dashed, line width=1.5pt]
plot[]
table[x index = 0, y index =2] from \datasetA;
\addlegendentry{$\gamma=1$}

\addplot [color=black, line width=1.5pt]
plot[]
table[x index = 0, y index =3] from \datasetA;
\addlegendentry{$\gamma=2$}

\addplot [color=black!50, line width=1.5pt]
plot[]
table[x index = 0, y index =4] from \datasetA;
\addlegendentry{$\gamma=3$}

\addplot [color=black!10, line width=1.5pt]
plot[]
table[x index = 0, y index =5] from \datasetA;
\addlegendentry{$\gamma=4$}

\end{axis}

\end{tikzpicture}
\caption{\edit{Delay vs traffic load $\bar{a}$ for various values of $\gamma$ (feedback delay $d=100$, packet loss rate $0.2$ so the maximum feasible traffic load is $0.8$). }}\label{Fig:gamma}
\end{figure}


\subsection{Impact of Imperfect Prediction: Rate Sub-Optimality}

Transmission policies $P$ use estimator $\hat{\theta}(\cdot)$ to make a $d$-step ahead prediction of $Q^r_k$.    As discussed in more detail later, use of prediction lowers delay.   However, inevitably, this estimator will make mistakes when predicting $Q^r_k$ due to the uncertainty in the fate of the packets ``in flight'', i.e. those transmitted but not yet acknowledged.     When $\hat{Q}^r_k\ge \gamma$ and $Q^r_k=0$ then the scheduler will send extra coded packets that are not useful (since $Q^r_k=0$ there are no outstanding losses at the receiver).   Prediction errors therefore translate into a loss in capacity, since these extra coded packets replace information packets that would have otherwise have been sent.

\subsubsection{Capacity Achieving Transmission Policies}

\edit{
We begin introducing the following technical Lemma,

\begin{lemma}\label{lem:two}
	{Suppose $E[{S}_k|\hat{Q}^r_{k}]\ge\epsilon>0$ when $\hat{Q}^r_{k}\le \gamma$. Then $\gamma-(1-p)\le E[\hat{Q}^r_{k+1}]\le \gamma+p$.}
\end{lemma}

\begin{proof}
	Recall (\ref{eq:update}),
	\begin{align}
	\hat{Q}^r_{k+1} = & \hat{Q}^r_{k}+(X_{k-d}-p)(S_{k-d} + C_{k-d}) \notag\\
	&+ {S}_kp -C_k(1-p)
	\end{align}
	\noindent when ${Q}^r_{k-d}\ge 1$ and (\ref{eq:update2}),
	\begin{align}
	\hat{Q}^r_{k+1} = & \hat{Q}^r_{k}+(X_{k-d}-p)(S_{k-d}+C_{k-d}) \notag \\
	& + S_kp-C_k(1-p) + C_{k-d}(1-X_{k-d}) 
	\end{align}
	\noindent when ${Q}^r_{k-d}< 1$. 
	Since 	$C_{k-d}(1-X_{k-d}) \ge0$ it follows that:
	\begin{align}
	\hat{Q}^r_{k+1}\ge &  \hat{Q}^r_{k}+(X_{k-d}-p)(S_{k-d}+C_{k-d}) \notag \\
	& + S_kp-C_k(1-p) 
	\end{align}
	\noindent for all values of ${Q}^r_{k-d}$ and taking expectations with respect to the packet arrival and loss processes,
	\begin{align}
	E[\hat{Q}^r_{k+1}| \hat{Q}^r_{k}]\ge \hat{Q}^r_{k} + E[{S}_k|\hat{Q}^r_{k}]p -E[C_k|\hat{Q}^r_{k}](1-p) 
	\end{align}
	\noindent where we have used the fact that $X_{k-d}$ is independent of $S_{k-d}$ and $C_{k-d}$.  We proceed by considering the following two cases.
	
	Case (i): $-\gamma+\hat{Q}^r_{k} \le 0$. Then $C_k=0$, and ${S}_k\in{0,1}$:
	\begin{align}
	-\gamma+E[\hat{Q}^r_{k+1}| \hat{Q}^r_{k}]\ge -\gamma+\hat{Q}^r_{k} + E[{S}_k|\hat{Q}^r_{k}]p  
	\end{align}
	\noindent Since $E[{S}_k|\hat{Q}^r_{k}]\ge \epsilon>0$, then it follows that $-\gamma+E[\hat{Q}^r_{k+1}| \hat{Q}^r_{k}]$ is strictly increasing and $-\gamma+\hat{Q}^r_{k}\le E[{S}_k|\hat{Q}^r_{k}]p\le p$.
	
	Case (ii):$-\gamma+\hat{Q}^r_{k}> 0$. Then $C_k=1$, $S_k=0$, and:
	\begin{align}
	-\gamma+E[\hat{Q}^r_{k+1}| \hat{Q}^r_{k}]\ge -\gamma+\hat{Q}^r_{k} -(1-p) 
	\end{align}
	\noindent and $-\gamma+E[\hat{Q}^r_{k+1}| \hat{Q}^r_{k}]$ is strictly decreasing and $-\gamma+E[\hat{Q}^r_{k+1}| \hat{Q}^r_{k}]> -(1-p)$. 
	
	Hence, when $-\gamma+E[\hat{Q}^r_{k+1}| \hat{Q}^r_{k}]$ is outside the interval $[-(1-p),p]$ then is is strictly attracted to this interval, and once it is within it, it stays there. The latter holds regardless of the values of $\hat{Q}^r_{k}$ and so we have that $-(1-p)\le -\gamma+E[\hat{Q}^r_{k+1}]\le p$ as claimed.
\end{proof}
}
The following theorem bounds the capacity loss induced by prediction errors:

\begin{theorem}\label{lem:three}
	Suppose {estimate $\hat{Q}^r_k$ satisfies $|\hat{Q}^r_k - Q^r_k| \le \delta$ for some $\delta\ge 0$} and $E[{S}_k|\hat{Q}^r_{k}]\ge\epsilon>0$ when $\hat{Q}^r_{k}\le \gamma$.  Then $E[\hat{s}_k] \ge (1-p)-(\frac{1}{2}+\delta+(1+\delta)(1-p))/\gamma$as $k\rightarrow\infty$, where $\hat{s}_k:=\frac{1}{k}\sum_{i=1}^k\hat{S}_i$ and $S_k \le \hat{S}_k = 1-C_k$.
\begin{proof}
	We have that
	
	\begin{align}
	(Q^r_{k+1})^2& =  ([Q^r_k + {S}_k-(1-X_k)]^+)^2\\
	& \le ([Q^r_k + \hat{S}_k-(1-X_k)]^+)^2\\
	&\le (Q^r_k + \hat{S}_k-(1-X_k))^2\\
	&=(Q^r_k)^2 + 2(\hat{S}_k-(1-X_k))Q^r_k\notag\\
	&\qquad + (\hat{S}_k-(1-X_k))^2
	\end{align}
	
	Applying this recursively we have, $(Q^r_{k+1})^2\le (Q^r_1)^2 + 2\sum_{i=1}^k(\hat{S}_i-(1-X_i))Q^r_i+ \sum_{i=1}^k (\hat{S}_i-(1-X_i))^2$.  That is,
	
	\begin{align}
	\frac{1}{k}\sum_{i=1}^k(\hat{S}_i-(1-X_i))Q^r_i &\ge -\eta -\frac{1}{2}
	\end{align}
	
	\noindent since $0\le(\hat{S}_i-(1-X_i))^2\le 1$, where $\eta:=\frac{1}{2k}(Q^r_1)^2$. Adding and subtracting $\gamma(\frac{1}{k}\sum_{i=1}^k\hat{S}_i-(1-p))$ to the LHS yields:
	
	\begin{align}
	&\frac{1}{k}\sum_{i=1}^k((-\gamma +Q^r_i)\hat{S}_i-(1-X_i)(Q^r_i-\gamma))  \notag \\
	&\qquad+\gamma  (\hat{s}_k-(1-\overline{x}_k)) \ge -\eta-\frac{1}{2}  \label{eq:inequal1}
	\end{align}
	
	\noindent where $\overline{x}_k:=\frac{1}{k}\sum_{i=1}^k{X}_i$.   
	
	The scheduler selects $\hat{S}_i \in\arg\min_{S\in\{0,1\}} (-\gamma +\hat{Q}^r_i)S$.  When $\hat{Q}^r_i\ge \gamma$ then $\hat{S}_i=0$ and otherwise $\hat{S}_i=1$.  Hence, $(-\gamma +\hat{Q}^r_i)\hat{S}_i\le 0$ for all $i=1,2,\dots$ and so $\frac{1}{k}\sum_{i=1}^k(-\gamma +\hat{Q}^r_i)S_i\le 0$.   Since $|\hat{Q}^r_k-Q^r_k|\le \delta$ and $\hat{S}_i\in \{0,1\}$ then $-\frac{1}{k}\sum_{i=1}^k(-\gamma +{Q}^r_i)\hat{S}_i\ge -\delta$.  Combining this with (\ref{eq:inequal1}) and taking expectation over the loss and arrival processes yields

	\begin{align}
	&\gamma  (E[\hat{s}_i]-(1-p) \ge \frac{1}{k}\sum_{i=1}^k(1-p)(E[Q^r_i]-\gamma) -\eta-\frac{1}{2}-\delta \label{eq:inequal2}
	\end{align}
	
	\noindent where we have used the fact that $X_i$ is independent of $Q^r_i$.  By Lemma \ref{lem:two} and the fact that  $|\hat{Q}^r_k-Q^r_k|\le \delta$ we have that $E[Q_i^r]-\gamma\ge -1-\delta$.    Combining this with (\ref{eq:inequal2}) yields
	
	\begin{align}
	&\gamma  (E[{s}_i]-(1-p))	\ge -(1+\delta)(1-p) -\eta-\delta-\frac{1}{2}
	\end{align}
	
	\noindent and the claimed result now follows by rearranging and using the fact that  $\eta\rightarrow0$ as $k\rightarrow\infty$.
\end{proof}
\end{theorem}

Theorem \ref{lem:three} says that as parameter $\gamma\rightarrow\infty$ the transmission slots $E[\hat{S}_k]$ available for sending information packets tends to the path capacity $1-p$.  That is, the transmission policy is capacity achieving as $\gamma\rightarrow\infty$.     The requirement that $E[{S}_k|\hat{Q}^r_{k}]\ge\epsilon>0$ when $\hat{Q}^r_{k}\le \gamma$ excludes transient arrival processes (for instance when packets arrive for a period of time and then no further arrivals happen) and is satisfied when, for example, the packet arrival process is ergodic and independent of the receiver queue.
\color{black}

\subsubsection{Estimating Rate Sub-Optimality For Small $\gamma$}

Lemma~\ref{lem:three} tells us that for $\gamma$ large enough our scheduler is achieving, even where the feedback delay, $d$, is greater than zero. However, this is not as comforting as it might seem at first sight since Theorem \ref{lem:one} also tells us that large $\gamma$ can lead to a large receiver queue and so large decoding delays.   By taking a different analysis approach, however, we can obtain fairly good estimates of the capacity loss induced by prediction errors when $\gamma$ is small.   These estimates indicate that the capacity loss is moderate.

Recall that under transmission policy $P$, when $\hat{Q}^r_k>\gamma$ then a coded packet is transmitted.   However, if $Q^r_k= 0$ then this coded packet is not useful, since there is no outstanding receiver queue i.e. no outstanding packet loss that bound benefit from the coded packet.   Defining random variable  $R_k$ which takes value 1 when $\hat{Q}^r_k>\gamma$ and  $Q^r_k=0$ and 0 otherwise then $\bar{r} = \frac{1}{K} \lim_{K\rightarrow\infty}\sum_{k=1}^{K}R_k$ is the transmission rate for redundant coded packets.   We would like to estimate $\bar{r}$.   

To proceed we make the following simplifying assumptions: (i) feedback is only received every $d$ slots, (ii) either an information packet or a coded packet is transmitted in every slot, (iii) $\gamma=0$,  and (iv) ${Q}^r_{k}=0$ for $k=id+1$, $i=0,1,\dots$. The assumptions mean that we have less knowledge of the decoder status and so expect the number of redundant packets transmitted to be larger i.e. we expect our estimate of $\overline{r}$ to be larger than the true value. Also, assumption (iv) implies that after $d$ slots we 

By assumption (i), at slots $k=id+1$, $i=0,1,\dots$ we perform update
\begin{align}
\hat{Q}^r_{k-d}&={Q}^r_{k-d}+\sum_{j=k-d}^{k-1}({S}_{j}p-C_j(1-p))\\
&\stackrel{(a)}{=}{Q}^r_{k-d}+dp-\sum_{j=k-d}^{k-1}C_j\le {Q}^r_{k-d}+dp
\end{align}
where in step $(a)$ we have used assumption (ii) that $S_j+C_j=1$.   By assumptions (iii) and (iv), over the next $d$ slots $\{k+1,\dots,k+d\}$ then at most $\lceil dp\rceil$ coded packets will be sent (fewer packets may be sent depending on the sample path $S_j$, $C_j$, $j\in\{k-d,\dots,k\}$ and when the threshold $\gamma>0$).  Letting $U_k=\sum_{j=k-d}^{k-1}{S}_{j}X_j$ denote the number of erased information packets then the number of redundant coded packets transmitted over slots $\{k+1,\dots,k+d\}$ is upper bounded by $\max\{0,\lceil{dp}\rceil-U_k\}$.     Letting $a_k=\sum_{j=k-d}^{k-1}{S}_{j}$ denote the number of information packets sent over slots $\{k-d,\dots,k\}$ then  $U_k$ is distributed as $\text{Prob}(U_k=u)=\mathcal{B} \left(d a_k, p, u\right)$, where $\mathcal{B}$ is the Binomial distribution.   Approximating $da_k$ by $\lfloor d\bar{a}\rfloor$ then an estimate of the number of redundant coded packets transmitted over interval $\{k+1,\dots,k+d\}$, normalised by the interval duration $d$, is
\begin{align}
\hat{r} = \sum_{u=0}^{\lceil{dp}\rceil } \frac{ \lceil{dp}\rceil  - u}{d} \mathcal{B} \left( \lfloor d \cdot \bar{a} \rfloor, p, u\right)\label{eq:estdummy}
\end{align}

Despite the assumptions made in deriving $(\ref{eq:estdummy})$, empirical tests indicate that the estimator is nevertheless quite accurate.   For example, Fig.~\ref{Fig: ProbDummyPackets} compares estimate $\hat{r}$ with the measured average number of redundant packets transmitted $\bar{r}$ as the feedback delay and loss rate $p$ is varied. It can be seen that $\hat{r}$ is essentially an upper bound on $\bar{r}$, and while it becomes less accurate as the loss rate $p$ increases it stays reasonably close to the true value.  Observe also that in Fig~\ref{Fig: ProbDummyPackets_c} the mean rate $a=0.6$ of packet arrivals is significantly less than the path capacity $1-p=0.9$ and so assumption (ii) (persistent queue backlog at the transmitter) is violated, nevertheless the estimate $\hat{r}$ remains accurate.

These results in Fig.~\ref{Fig: ProbDummyPackets} indicate that the capacity loss due to the transmission of redundant packets generally stays below $5\%$.   However, observe also that when the delay $d$ is less than the reciprocal of the loss rate $1/p$ then no redundant packets are sent i.e. $\bar{r}=0$.  This behaviour is accurately captured by $\hat{r}$ (since $\lceil dp \rceil=1$ when $d<1/p$ in (\ref{eq:estdummy})).   Hence, on links with lower loss larger feedback delays can be tolerated without incurring redundant packet transmissions e.g. for $p=0.01$ (a typical path loss rate in the Internet) feedback delays of up to 100 slots yield $\bar{r}=0$.

\begin{figure}
	\def \figurewidth {0.8\columnwidth}
	\def \figureheight {3.0cm}
	\subfloat[][$a=0.9$, $p=0.1$, $\gamma=1$]{\pgfplotstableread{figures/data/prob_dummy_packets_p01_a09.dat}{\dataset} 
\begin{tikzpicture} 
font = \scriptsize, 
\begin{axis}[scale only axis, 
width=\figurewidth,
height=\figureheight,
mark options={solid},
ymin=0,
ymax=0.1,
xmin=0,
xmax=100,
ytick={0,0.05,0.1},
yticklabels={0,0.05,0.1},
ylabel={rate($\hat{r},\overline{r}$)},
xlabel={$d$ (\# slots)},
compat=1.3,
legend style={legend pos=north east, draw=none, fill=none,legend columns=1, font = \footnotesize,legend cell align=left}
]

\addplot [dotted,line width=1.5pt]
plot[]
table[x index = 0, y index =1] from \dataset;
\addlegendentry{$\overline{r}$} 

\addplot [solid,line width=1.5pt]
plot[]
table[x index = 0, y index =2] from \dataset;
\addlegendentry{$\hat{r}$} 

\end{axis}

\end{tikzpicture}} \\
 	\subfloat[][$a=0.8$, $p=0.2$, $\gamma=1$]{\pgfplotstableread{figures/data/prob_dummy_packets_p02_a08.dat}{\dataset} 
\begin{tikzpicture} 
font = \scriptsize, 
\begin{axis}[scale only axis, 
width=\figurewidth,
height=\figureheight,
mark options={solid},
ymin=0,
ymax=0.2,
xmin=0,
xmax=100,
ytick={0,0.05,0.1,0.15,0.2},
yticklabels={0,0.05,0.1, 0.15, 0.2},
ylabel={rate($\hat{r},\overline{r}$)},
xlabel={$d$ (\# slots)},
compat=1.3,
legend style={legend pos=north east, draw=none, fill=none,legend columns=1, font = \footnotesize,legend cell align=left}
]

\addplot [dotted,line width=1.5pt]
plot[]
table[x index = 0, y index =1] from \dataset;
\addlegendentry{$\overline{r}$}

\addplot [solid,line width=1.5pt]
plot[]
table[x index = 0, y index =2] from \dataset;
\addlegendentry{$\hat{r}$}

\end{axis}

\end{tikzpicture}} \\
	\subfloat[][$a=0.6$, $p=0.1$, $\gamma=1$]{\pgfplotstableread{figures/data/prob_dummy_packets_p01_a06.dat}{\dataset} 
\begin{tikzpicture} 
font = \scriptsize, 
\begin{axis}[scale only axis, 
width=\figurewidth,
height=\figureheight,
mark options={solid},
ymin=0,
ymax=0.1,
xmin=0,
xmax=100,
ytick={0,0.05,0.1},
yticklabels={0,0.05,0.1},
ylabel={rate($\hat{r},\overline{r}$)},
xlabel={$d$ (\# slots)},
compat=1.3,
legend style={legend pos=north east, draw=none, fill=none,legend columns=1, font = \footnotesize,legend cell align=left}
]

\addplot [dotted,line width=1.5pt]
plot[]
table[x index = 0, y index =1] from \dataset;
\addlegendentry{$\overline{r}$} 

\addplot [solid,line width=1.5pt]
plot[]
table[x index = 0, y index =2] from \dataset;
\addlegendentry{$\hat{r}$} 

\end{axis}

\end{tikzpicture}\label{Fig: ProbDummyPackets_c}}
	\caption{Rate of transmitting dummy packets, simulation results and analytic estimate. }
	\label{Fig: ProbDummyPackets}
\end{figure}
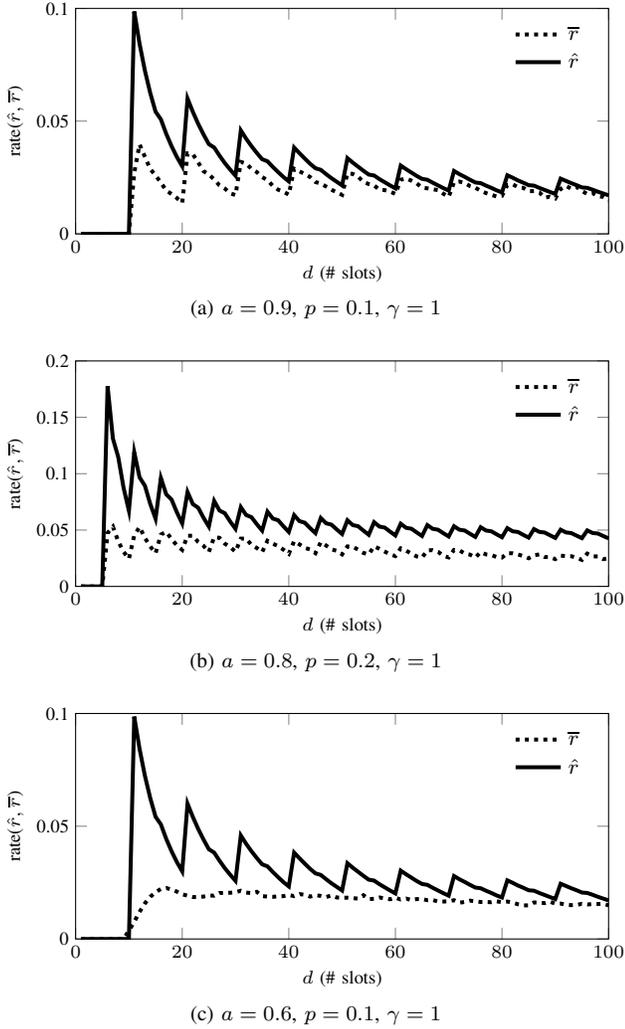



\subsection{Performance Evaluation}

We conclude this section by comparing the performance of the transmission policies introduced here with state-of-the-art alternatives, namely (i) ARQ, (ii) random linear block codes and (iii) the low delay code construction from~\cite{Karzan2017}. Note that the latter two are open-loop approaches, i.e. do not make use of feedback to trigger transmission of extra packets. Fig.~\ref{Fig: ComparisionDelayCodingPolicies2} plots the measured end-to-end delay Vs. the feedback delay for these schemes. As expected, for the block code and low delay coding schemes the end-to-end delay is constant, and does not vary with the feedback delay, also the end-to-end delay with the low delay code construction is around half of that for the block code (which is consistent with the results reported in~\cite{Karzan2017}). It can be seen that with ARQ the end-to-end delay increases linearly with the feedback delay $d$, and for delays greater than 40 slots the end-to-end delay with ARQ is larger than with any of the other approaches. However, for lower feedback delays the end-to-end delay with ARQ is lower than for the two open-loop coding schemes. The class of transmission policies introduced here provides a balance between ARQ and the low delay code construction.  Namely, when the feedback delay is low its end-to-end delay performance is similar to that of ARQ (which is known to be delay optimal when the feedback delay is zero) and as the feedback delay becomes large its end-to-end performance is similar to that of the low delay code construction.   For intermediate values of feedback delay, the proposed class of transmission policies offers lower end-to-end delay than any of the competing approaches.

\begin{figure}
	\def \figurewidth {0.8\columnwidth}
	\def \figureheight {3.0cm}
	\pgfplotstableread{figures/data/comparison_delayVsdelayLink_a07_p02.dat}{\dataset} 
\begin{tikzpicture}
font = \scriptsize, 
\begin{axis}[scale only axis, 
width=\figurewidth,
height=\figureheight,
mark options={solid},
ymin=0,
ymax=120,
xmin=0,
xmax=100,
ylabel={e2e delay (\# slots)},
xlabel={$d$ (\# slots)},
compat=1.3,
legend style={legend pos = north west, draw=none, fill=none , legend columns=2, font = \scriptsize}
]

\addplot [color=black, line width = 1.5pt, dashed]
plot [error bars/.cd, y dir = both, y explicit]
table[x index = 0, y index =1, y error plus index=2, y error minus index=2] from \dataset;
\addlegendentry{\emph{ARQ}}

\addplot [color=black, line width = 1.5pt]
plot [error bars/.cd, y dir = both, y explicit]
table[x index = 0, y index =3, y error plus index=4, y error minus index=4] from \dataset;
\addlegendentry{\emph{Block}}

\addplot [color=black!60, line width = 1.5pt]
plot [error bars/.cd, y dir = both, y explicit]
table[x index = 0, y index =5, y error plus index=6, y error minus index=6] from \dataset;
\addlegendentry{\emph{Low delay}}

\addplot [color=black!30, line width = 1.5pt]
plot [error bars/.cd, y dir = both, y explicit]
table[x index = 0, y index =7, y error plus index=8, y error minus index=8] from \dataset;
\addlegendentry{\emph{Proposal}}

\end{axis}

\end{tikzpicture}
	\caption{End-to-end delay vs feedback delay ($d$) for ARQ and various coding approaches. Configuration parameters are $a=0.7$, $p=0.2$, $k=50$ and $\gamma = 1$, block size $50$.  The data is for $10^4$ slots repeated $100$ times and the figure also shows the 95\% confidence intervals.}
	\label{Fig: ComparisionDelayCodingPolicies2}
\end{figure}
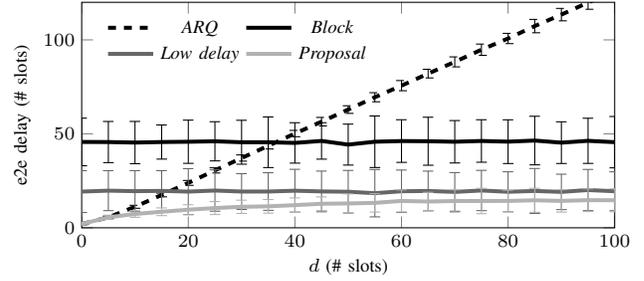
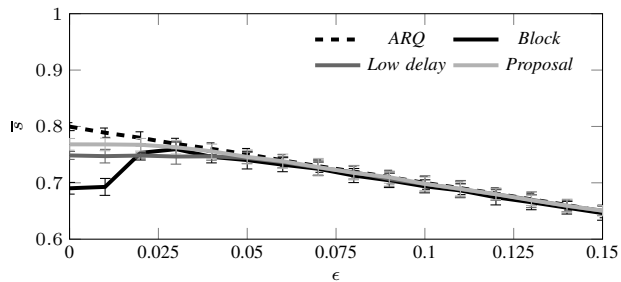
\begin{figure}
	\def \figurewidth {0.8\columnwidth}
	\def \figureheight {3.0cm}
	\pgfplotstableread{figures/data/comparison_perfromance_p02_d100.dat}{\dataset} 
\begin{tikzpicture}
font = \scriptsize, 
\begin{axis}[scale only axis, 
width=\figurewidth,
height=\figureheight,
mark options={solid},
ymin=0.6,
ymax=1,
xmin=0,
xmax=0.15,
xtick={0,0.025,0.05,0.075,0.1,0.125,0.15},
xticklabels={0,0.025,0.05,0.075,0.1,0.125,0.15},
ylabel={$\overline{s}$},
xlabel={$\epsilon$},
compat=1.3,
legend style={legend pos = north east, draw=none, fill=none , legend columns=2, font = \scriptsize}
]

\addplot [color=black, line width = 1.5pt, dashed]
plot [error bars/.cd, y dir = both, y explicit]
table[x index = 0, y index =1, y error plus index=2, y error minus index=2] from \dataset;
\addlegendentry{\emph{ARQ}}

\addplot [color=black, line width = 1.5pt]
plot [error bars/.cd, y dir = both, y explicit]
table[x index = 0, y index =3, y error plus index=4, y error minus index=4] from \dataset;
\addlegendentry{\emph{Block}}

\addplot [color=black!60, line width = 1.5pt]
plot [error bars/.cd, y dir = both, y explicit]
table[x index = 0, y index =5, y error plus index=6, y error minus index=6] from \dataset;
\addlegendentry{\emph{Low delay}}

\addplot [color=black!30, line width = 1.5pt]
plot [error bars/.cd, y dir = both, y explicit]
table[x index = 0, y index =7, y error plus index=8, y error minus index=8] from \dataset;
\addlegendentry{\emph{Proposal}}

\end{axis}

\end{tikzpicture}
	\caption{Achieved transmission rate $\bar{s}$ of information packets as the arrival rate approaches capacity, $\bar{a}=1-p-\epsilon$. Configuration parameters are, $p=0.2$, $k=50$, $d=100$, $\gamma = 1$, block size $50$. The data is for $10^4$ slots repeated $100$ times and the figure also shows the 95\% confidence intervals}
	\label{Fig: ComparisonPerformance}
\end{figure}

As noted above, the use of prediction introduces a trade-off between delay and rate, since prediction errors lead to transmission of redundant coded packets.   In comparison ARQ, which is purely reactive and involves no prediction, is capacity achieving but at the cost of increased end-to-end delay compared to when prediction is used (see Fig.~\ref{Fig: ComparisionDelayCodingPolicies2}).  The trade-off between delay and rate seems like a fundamental one since predictions allow lower delay to be achieved, but prediction errors are inevitable when losses are stochastic.   Fig.~\ref{Fig: ComparisonPerformance} explores this trade-off in more detail.  For a fixed packet loss rate $p$ this figure plots the achieved transmission rate $\bar{s}$ of information packets as the arrival rate $\bar{a}$ approaches capacity, namely $\bar{a}=1-p-\epsilon$ where $\epsilon$ is indicated on the x-axis of the plot. It can be seen that when $\epsilon \gg 0$, the achieved transmission rate equals the arrival rate for  all four schemes.  However, as the arrival rate approaches capacity ($\epsilon\rightarrow 0$) the achieved transmission rate falls below the arrival rate for all schemes apart from ARQ.  Since we use a fixed block size, the block code is not capacity achieving and this behaviour is to be expected.  Similarly, the low delay code construction incurs an overhead at the end of a connection.   Interestingly, observe that the achieved transmission rate with the transmission policy introduced here is higher than either of these schemes, that is the loss in capacity due to redundant coded transmissions is lower.

	\section{Generalising to Networks of Flows }

\edit{In this section we first make the observation that the policy $P$ introduced in Section \ref{subsec:transmissionpolicies} can also be seen as an approximate dual-subgradient update for an associated convex optimisation problem. This connection allows us to  exploit convex optimisation results to extend policy $P$ to networks with multiple flows sharing multiple lossy paths. We illustrate this using a simple multipath example.}

\subsection{Relating Policy $P$ with Convex Optimization}

Assuming, for simplicity, that the arrival queue $Q^t$ is persistently backlogged, then transmission policy $P$ is,
\begin{align}
&C_k \in \arg\min_{C\in\{0,1\}}  (-\hat{Q}^r_k +\gamma ) C\\
&Q^r_{k+1} = [Q^r_{k} + {S}_k \cdot X_k - C_k(1-X_k)]^+\\
&\hat{Q}^r_{k}={Q}^r_{k-d}+\sum_{j=k-d}^{k-1}({S}_{j}p-C_j(1-p))\\
&S_k=1-C_k
\end{align}
This is a natural threshold-based policy, namely a coded packet is sent whenever the virtual receiver queue is larger than $\gamma$, combined with use of predictor $\hat{Q}^r_{k}$ to mitigate the impact of the feedback delay $d$. 

Now, consider the convex optimisation problem $C$,	
\begin{align}
&\min_{s\in [0,1]}  - s \\
&s \le 1-p
\end{align}

\noindent where it will be helpful to think of $s$ as the average transmission rate of information packets.  Letting $c=1-s$ (which can be thought of as the average transmission rate of coded packets) the constraint $s \le 1-p$ ensures that $sp\le c(1-p)$, so enough coded packets are sent to recover from packet losses. This optimisation has the trivial solution $s=1-p$, but this is not where our interest lies. Rather, we focus on the relationship between this optimisation and policy $P$.  

Optimisation $C$ is convex and, provided $0\le p<1$, then the interior of the feasible set is non-empty, i.e. the Slater condition is satisfied and so strong duality holds. The Lagrangian is $L(s,\lambda):=- s+\lambda(s-(1-p))$ and the standard dual subgradient update for solving the optimisation is:
\begin{align}
&s_{k}\in \arg\min_{s\in [0,1]}  L(s,\lambda_k) \stackrel{(a)}{=}\arg\min_{s\in \{0,1\}}  (-1 +\lambda_k)s \\
&\lambda_{k+1}=[\lambda_k+\frac{1}{\gamma}(s_{k}-(1-p))]^+
\end{align}
\noindent where step size $1/\gamma>0$ and equality $(a)$ follows by dropping the terms in $L(s,\lambda)$ that do not depend on $s$ (and so do not affect the $s$ that minimises $L(s,\lambda)$), and noting that the solution must lie at an extreme point, i.e. $0$ or $1$.   

Defining $q_k=\gamma\lambda_k$, then this dual subgradient update can be rewritten equivalently as:
\begin{align}
&c_{k}\in \arg\min_{c\in \{0,1\}}  (-q_k+\gamma)c \\
&q_{k+1}=[q_{k}+s_{k}p-c_{k}(1-p)]^+\\
&s_{k}=1-c_k
\end{align}

The similarity of this update with transmission policy $P$ is immediately apparent, including the fact that $s_k$ and $c_k$ are $\{0,1\}$ valued.  However, it can be seen that there are also some important differences: (i) the average loss rate $p$ is used, rather than the loss rate process $\{X_k\}$; (ii) scaled multiplier $q_k$ is a real-valued quantity, whereas packet queue $Q_k$ is integer valued; (iii) feedback delay $d$ is ignored. Nevertheless, despite these differences, recent results on approximate convex optimisation in \cite{Valls2017} can be used to establish a strong connection between the update generated by policy $P$ and the optimal solution to problem $C$.   

Letting $\epsilon_k = (\hat{Q}^r_{k} -Q^r_{k})/\gamma$ and $\delta_k=({S}_k \cdot X_k - C_k(1-X_k)) - ({S}_kp  - C_k(1-p)) = X_k-p$, we can write transmission policy $P$ equivalently as:
\begin{align}
&C_k \in \arg\min_{C\in\{0,1\}}  (-\hat{Q}^r_k +\gamma ) C\label{eq:reform1}\\
&{Q}^r_{k+1} =[{Q}^r_k +{S}_kp  - C_k(1-p)+\delta_k]\\
&\hat{Q}^r_k = {Q}^r_k+\gamma\epsilon_k\\
&S_k=1-C_k\label{eq:reform2}
\end{align}
We now recall the following, which corresponds to \cite[Theorem 1]{Valls2017},  
\begin{theorem}
\label{th:perturbed_theorem}
Consider the convex optimisation: $\min_{x\in X} f(x)\quad s.t. \; g(x)+\delta\le0,\ j=1,\dots,m$ where $f:X\rightarrow R$, $g:X\rightarrow R^m$ are convex functions, $X$ is a bounded convex subset of $R^n$ and $\delta\in R^m$.   Let dual function $h(\lambda,\delta):=\inf_{x\in X}f(x)+\lambda^T(g(x)+\delta)$ and consider the update
\begin{align}
\lambda_{k+1} & =  [\lambda_{k} + \alpha \partial h(\mu_k, \delta_k) ]^+ \label{eq:th_l_udpate}
\end{align}
where $\mu_{k} = \lambda_k + \epsilon_k$ with  $\lambda_1 \in R^m_+$ and $\{ \epsilon_k \}$ a sequence of points from $R^m$ such that $\mu_k \succeq 0$ for all $k$.
Suppose the Slater condition is satisfied and that $ \delta_k $ is an ergodic stochastic process with expected value $\delta$ and $E(\| \delta_k - \delta \|_2^2) = \sigma_\delta^2$ for some finite $\sigma_\delta^2$.
Further, suppose that $ \frac{1}{k} \sum_{i=1}^k \| \epsilon_i \|_2 \le \epsilon$ for all $k$ and some $\epsilon \ge 0 $. 
Then, 
\begin{flalign*}
 \textup{(i)}  & \quad
\lim_{k \to \infty}| E ( f(\bar x_k) - f^\star(\delta) ) | 
\le  \frac{\alpha M}{2}  + 2 \epsilon \sigma_g  \\
 \textup{(ii)}  & \quad
 \lim_{k \to \infty}   E \left( g(\bar x_k) + \delta \right)  \preceq 0  \\
  \textup{(iii)}  & \quad E\left(
   \frac{1}{k} \sum_{i=1}^k \lambda_i \right)   \prec \infty \qquad k=1,2,\dots 
\end{flalign*}
where $\bar x_k = \frac{1}{k} \sum_{i=1}^k x_i$ and $M = \sigma^2_g + \sigma_\delta^2$.
\end{theorem}
Applying this to optimisation $C$ and identifying (\ref{eq:reform1})-(\ref{eq:reform2}) as the perturbed update (\ref{eq:th_l_udpate}), we obtain the following:
\begin{lemma}\label{lem:convex}
Suppose loss process $\{X_k\}$ is ergodic. Then, under policy $P$, we have $\lim_{k\rightarrow\infty}|E[\bar{S}_k]-(1-p)| \le \frac{1+2d}{\gamma}$, where $\bar{S}_k:=\frac{1}{k}\sum_{i=1}^k {S}_i$, and $E\left(\frac{1}{k} \sum_{i=1}^k Q^r_i \right)<\infty$ for all $k$.
\end{lemma}
\begin{proof}
To apply Theorem \ref{th:perturbed_theorem} we need to show for $\epsilon_k = (\hat{Q}^r_{k} -Q^r_{k})/\gamma$ and $\delta_k= X_k-p$ that: (i)$ \frac{1}{k} \sum_{i=1}^k \| \epsilon_i \|_2 \le \epsilon$; and (ii) $\delta_k$ has finite variance $\sigma_\delta^2\le 1$. Observing that $\gamma|\hat{Q}^r_{k} -Q^r_{k}|\le \gamma d$, then (i) follows immediately with $\epsilon= d/\gamma$. Since $X_k\in\{0,1\}$ then $|\delta_k|\le 1$ and it follows immediately that $\delta_k$ has finite variance $\sigma_\delta^2\le 1$.
\end{proof}
Observe that the bound on $E[\bar{S}_k]$ in Lemma \ref{lem:convex} is not particularly useful when $\gamma$ is small, nor the bound on $E\left(\frac{1}{k} \sum_{i=1}^k Q^r_i \right)$. We nonetheless previously showed that much tighter bounds can be derived for policy $P$. On the other hand, the approach used to derive Lemma \ref{lem:convex} is indeed rather interesting, because it can be used to show that transmission policy $P$ can be embedded as a building block within solution updates for general convex optimisation problems, and not only will behave sensibly, but can be analysed via Theorem \ref{th:perturbed_theorem} and related results from the area of approximate convex optimisation. We illustrate this in more detail in the next section using multipath communications as an illustrative example.

%
%

\subsection{Example: Multipath Transmission}


\edit{
Consider the following lossy network multi-commodity flow setup, which is a variant of the setup in \cite{lossy09}.  Let $G=(V,E)$ denote a graph with vertices $V$ and edges $E\subset V\times V$.  Time is slotted, edges have unit capacity and $p_e$ denotes the packet loss rate on edge $e$, with losses being i.i.d. across slots.  The network carries a set $F\subset V\times V$ of flows, with flow $f=(s,d)\in F$ having source/transmitter $s$ and destination/receiver $d$.  Each source $s$ has a single destination\footnote{\edit{It may be possible to generalise this unicast setup to multicast, but we leave this as future work.}} $d$, but multiple paths may be used to transmit packets from each source to the corresponding destination.  

Let $P_f\subset 2^E$ denote the set of usable paths between source $s$ and destination $d$.  In general $P_f$ will be a subset of all possible paths from $s$ to $d$, determined by delay requirements, routing protocols \emph{etc}.  We assume paths in $P_f$ have no loops and that the time taken to send a packet from source $s$ to destination $d$ is the same\footnote{\edit{The assumption that the delay on each path is the same can be readily relaxed. In the analysis it is used to allow transmission events (e.g. $S_k$ in the notation of the earlier part of this paper) to be referred to the receiver, which effectively lumps the forward and reverse path delays into the feedback delay.  When there are multiple paths with differing but known delays then this can still be achieved by using an earliest-deadline first policy to select which packet to send on a path when a transmission is made.}} for all paths in $P_f$.  Let $g_{i,e}$ denote the number of hops along path $i$ between the source and edge $e$, with  $g_{i,e}=\infty$ for edges not on path $i$.  Associate with path $i\in P_f$ the vector $a_{f,i}\in\mathbb{R}_+^{|E|}$ with element corresponding to edge $e$ equal to $1/\Pi_{e'\in E:g_{i,e'}<g_{i,e}} (1-p_{e'})$ when $e$ lies on path $i$ and otherwise set equal to $0$.   The elements of $a_{f,i}$ capture the accumulated packet loss along path $i$.  Let $r_{f,i}\in\mathbb{R}_+$ denote the rate at which packets are sent by flow $f$ along path $i\in P_f$.  Gathering the vectors $a_e$ into matrix $A\in\mathbb{R}_+^{|E|\times n}$ where $n:={\sum_{f\in F}|P_f|}$ is the number of network paths and rates $r_{f,i}$ into vector $r\in\mathbb{R}_+^{n}$ then for feasibility the flow rates must satisfy network capacity constraint
\begin{align}
    Ar\le 1
\end{align}
where $1\in\mathbb{R}^{|E|}$ denotes the vector with all elements 1 and the inequality being interpreted element-wise.

Now consider the optimisation
\begin{align}
    &\quad\min_{r\in\mathbb{R}_+^{n}: Ar\le 1} -\sum_{f\in F} s_f \label{eq:o1}\\
    &\text{s.t. } s_f \le \sum_{i\in P_f}(1-p_i)r_{f,i}\label{eq:con}
\end{align}
where $p_i:=1-\Pr_{e\in i}(1-p_e)$ is the aggregate loss rate along path $i$.   Think of $s_f$ as the aggregate rate at which flow $f$ sends information packets.  Letting $s_{f,i}$ denote the fraction of information packets sent by flow $f$ along path $i$ then $s_f=\sum_{i\in P_f}s_{f,i}$ and $c_{f,i}=r_{f,i}-s_{f,i}$ is the rate at which coded packets are sent along path $i$.   Constraint (\ref{eq:con}) ensures that
\begin{align}
\sum_{i\in P_f} p_{i}s_{f,i} \le \sum_{i\in P_f} (1-p_i)c_{f,i}
\end{align}
i.e. enough coded packets are sent to recover from packet losses.

The Lagrangian is $L(s,r,\lambda)= -\sum_{f\in F} s_f +\sum_{f\in F} \lambda_f(s_f - \sum_{i\in P_f}(1-p_{i})r_{f,i})$, and the Frank-Wolfe variant of the dual subgradient update is:

\begin{align}
&r_{k+1}\in \arg\min_{r\in\mathbb{R}_+^{n}:Ar\le 1}  \sum_{f\in F}\sum_{i\in P_f}(-Q_{f,k}(1-p_{i}))r_{f,i}\\
&s_{f,k+1}\in \arg\min_{0\le s\le \sum_{i\in P_f}r_{f,i}}  (-1/\alpha +Q_{f,k})s \\
&Q_{f,k+1}=[Q_{f,k}+s_{f,k}-\sum_{i\in P_f}(1-p_{i})r_{f,i}]^+
\end{align}
\noindent where $Q_{f,k}=\lambda_{f,k}/\alpha$ and $\alpha>0$ is a design parameter.  Since $s_{f,k+1}$ is the solution of a linear programme its value lies at an extreme point and so the updates can be equivalently replaced by $s_{f,k+1}\in \arg\min_{s\in \{0,\sum_{i\in P_f}r_{f,i}\}}  (-1/\alpha +Q_{f,k})s$.  Similarly, since $r_{k+1}$ is the solution of a linear programme is an extreme point of set $\{r\in \mathbb{R}_+^{n}:Ar\le 1\}$.   When $A$ is unimodular then the extreme points (and so $r_{k+1}$) are integer-valued.  More generally, we can always use randomised time-sharing to select a vector $R_{k+1}$ with $\{0,1\}$ valued elements such that $E[R_{k+1}]=r_{k+1}$.

Replacing $Q_{f,k}$ with virtual receiver queue $Q^r_{f,k}$ yields the update:

\begin{align}
&r_{k+1}\in \arg\min_{r\in\mathbb{R}_+^{n}:Ar\le 1}  \sum_{f\in F}\sum_{i\in P_f}(-Q_{f,k}(1-p_{i}))r_{f,i}\label{eq:r1}\\
&\text{Select }R_{k+1}\in\{0,1\}^{|E|}\text{ s.t. }E[R_{k+1}]=r_{k+1} \label{eq:r2}\\
&\hat{S}_{f,k+1}\in \arg\min_{S\in \{0,R_{f,k}\}}  (-1/\alpha +Q^r_{f,k})S \\
&Q^r_{f,k+1}=[Q^r_{f,k}+\hat{S}_{f,k}-\sum_{i\in P_f}(1-X_{i,k})R_{f,i}]^+\notag\\
&=[Q^r_{f,k}+\sum_{i\in P_f}\hat{S}_{f,i,k}X_{i,k}-\sum_{i\in P_f}(1-X_{i,k})C_{f,i,k}]^+
\end{align}
where $\hat{S}_f=\sum_{i\in P_f}\hat{S}_{f,i}$, $C_{f,i}=R_{f,i}-\hat{S}_{f,i}$ and $\{X_{i,k}\}$ are i.i.d random variables with $E[X_{i,k}]=p_i$ and with $X_{i,j}$ taking value 1 when a packet is erased on path $i$ in slot $k$ and $0$ otherwise.  By Theorem \ref{th:perturbed_theorem} this update converges to a ball around the solution of optimisation (\ref{eq:o1})-(\ref{eq:con}), with the size of the ball decreasing as scaling/step-size parameter $\alpha$ is decreased.

We can map this update onto the following physical setup. $\hat{S}_{f,k}$ is the number of information packets from flow $f$ to be transmitted in slot $k$. Note that $\hat{S}_{f,k+1}$ might take a value greater than 1, if multiple link transmit slots are available to flow $f$. Selecting $R_{f,i,k}=1$ corresponds to allocating transmission slot $k$ on link $i$ to a packet from flow $f$. When $\hat{S}_{f,k}=0$ (there is not an information packet to be sent) a coded packet is transmitted, $C_{f,i,k}=1$). The occupancy of Virtual receiver queue $Q^r_{f,k+1}$ increases when an information packet is lost, and decreases upon receiving a coded packet.  $\hat{S}_{f,k+1}$ is selected according to a threshold rule, namely non-zero when $Q^r_{f,k}<1/\alpha$ and a transmission slot is available (i.e. when $\sum_{i\in P_f} R_{f,i}>0$). When there is feedback delay we can replace $Q^r_{f,k}$ in this threshold rule with prediction $\hat{Q}^r_{f,k}$.

\begin{figure}
    \centering
    \begin{tikzpicture}

\node[rotate=12,scale=0.2,cloud, cloud puffs=9, aspect=2,cloud puff arc=120, minimum width=65 em, minimum height=52 em, align=center, draw, black!70,line width=0.03em] (operator) at (6.1,5.5) {};

\node (server) at (9,6) {\includegraphics[width=2.6em]{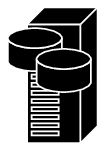}};
\node (pc) at (3,5) {\includegraphics[width=2.8em]{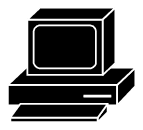}};

\node[anchor=north] at (server.south)  {$d$};
\node[anchor=north] at (pc.south)  {$s$};

\draw[thick] (pc.40) to[out=20,in=160] node[sloped,above,pos=0.4] {Path 1} (server.130);
\draw[thick] (pc.10) to[out=350,in=190] node[sloped,above,pos=0.5] {Path 2} (server.180);
\draw[thick] (pc.340) to[out=330,in=210] node[sloped,above,pos=0.6] {Path 3} (server.230);

\end{tikzpicture}
    \caption{\edit{Example multipath topology with three paths shared by three flows, each flow having the same source and destination.}}
    \label{fig:topo}
\end{figure}

We illustrate the application of this update to the simple multipath topology shown in Fig. \ref{fig:topo} which has three paths between source $s$ and destination $d$.  These paths are shared by three flows.  In this case update (\ref{eq:r1})-(\ref{eq:r2}) simplifies to element $f$ of $R_{i,k+1}$ taking value 1 (corresponding to transmitting a packet from flow $f$ on link $i$ in slot $k+1$) when the receiver backlog $Q^r_{f,k}(1-p_{i})$ for flow $f$ is the largest one amongst the three flows.  

Fig.~\ref{Fig: MultipathPerformance} shows the performance obtained as we vary the packet loss rate from 0.0 to 0.4 over the first path, and we keep the other two with a fixed erasure probability of $0.1$. Note that flows sharing the same path see the same loss probability. Fig.~\ref{Fig: MultipathPerformancea} shows the individual rates for each flow and it can be seen, the available capacity is equally shared between the flows.  Fig.~\ref{Fig: MultipathPerformanceb} shows the aggregate rate, which is obtained by summing the individual flow throughputs. It can be seen that the rate of the proposed multipath scheduler almost reaches the system capacity. 
}

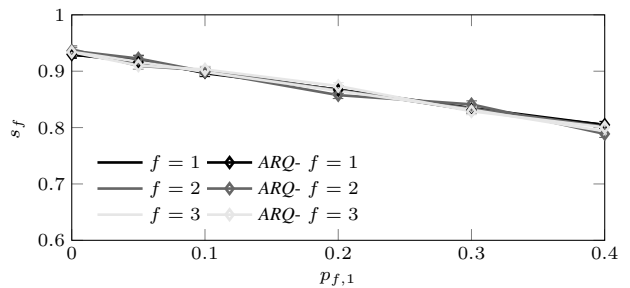
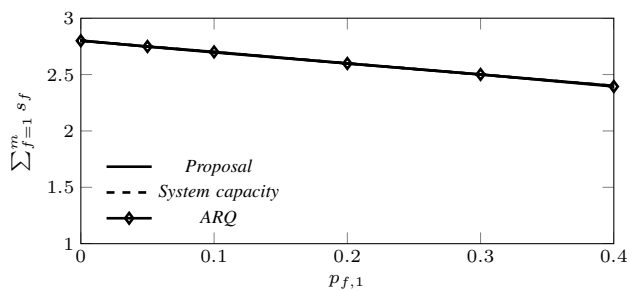
\begin{figure}
	\def \figurewidth {0.8\columnwidth}
	\def \figureheight {3.0cm}
	\subfloat[][Performance per flow]{\label{Fig: MultipathPerformancea}\pgfplotstableread{figures/data/performance_multipath_3.dat}{\dataset} 
\pgfplotstableread{figures/data/performance_multipath_arq_3.dat}{\datasetA} 
\begin{tikzpicture} 
font = \scriptsize,

\begin{axis}[scale only axis, 
width=\figurewidth,
height=\figureheight,
mark options={solid},
ymin=0.6,
ymax=1,
xmin=0,
xmax=0.4,
ylabel={$s_f$},
xlabel={$p_{f,1}$},
compat=1.3,
legend style={legend pos = south west, draw=none, fill=none , legend columns=2, font = \scriptsize}
]

\addplot [color=black, line width=1pt]
plot [error bars/.cd, y dir = both, y explicit]
table[x index = 0, y index =1, y error plus index=2, y error minus index=2] from \dataset;
\addlegendentry{$f=1$}

\addplot [color=black,mark=diamond, line width=1pt]
plot [error bars/.cd, y dir = both, y explicit]
table[x index = 0, y index =1, y error plus index=2, y error minus index=2] from \datasetA;
\addlegendentry{\emph{ARQ}- $f=1$}

\addplot [color=black!60, line width=1pt]
plot [error bars/.cd, y dir = both, y explicit]
table[x index = 0, y index =3, y error plus index=4, y error minus index=4] from \dataset;
\addlegendentry{$f=2$}

\addplot [color=black!60,mark=diamond, line width=1pt]
plot [error bars/.cd, y dir = both, y explicit]
table[x index = 0, y index =3, y error plus index=4, y error minus index=4] from \datasetA;
\addlegendentry{\emph{ARQ}- $f=2$}

\addplot [color=black!10, line width=1pt]
plot [error bars/.cd, y dir = both, y explicit]
table[x index = 0, y index =5, y error plus index=6, y error minus index=6] from \dataset;
\addlegendentry{$f=3$}

\addplot [color=black!10,mark=diamond, line width=1pt]
plot [error bars/.cd, y dir = both, y explicit]
table[x index = 0, y index =5, y error plus index=6, y error minus index=6] from \datasetA;
\addlegendentry{\emph{ARQ}- $f=3$}

\end{axis}

\end{tikzpicture}}\\
	\subfloat[][Aggregated System performance]{\label{Fig: MultipathPerformanceb}\pgfplotstableread{figures/data/performance_multipath_3.dat}{\dataset} 
\pgfplotstableread{figures/data/performance_multipath_arq_3.dat}{\datasetA} 

\begin{tikzpicture} 
font = \scriptsize,

\begin{axis}[scale only axis, 
width=\figurewidth,
height=\figureheight,
mark options={solid},
ymin=1.0,
ymax=3.0,
xmin=0,
xmax=0.4,
ylabel={$\sum_{f=1}^m s_f$},
xlabel={$p_{f,1}$},
compat=1.3,
legend style={legend pos = south west, draw=none, fill=none , legend columns=1, font = \scriptsize}
]

\addplot [color=black, line width=1pt]
plot [error bars/.cd, y dir = both, y explicit]
table[x index = 0, y index =7, y error plus index=8, y error minus index=8] from \dataset;
\addlegendentry{\emph{Proposal}}

\addplot [dashed,color=black, line width=1pt]
plot[]
table[x index = 0, y index =9] from \dataset;
\addlegendentry{\emph{System capacity}}

\addplot [color=black, mark=diamond ,line width=1pt]
plot [error bars/.cd, y dir = both, y explicit]
table[x index = 0, y index =7, y error plus index=8, y error minus index=8] from \datasetA;
\addlegendentry{\emph{ARQ}}

\end{axis}

\end{tikzpicture}}
	\caption{Performance of multipath system, with 3 information flows over 3 different paths. Feedback delay $d=10$.  The results are averaged after running the experiment during $10^5$ slots and the figure also shows the 95\% confidence intervals}
	\label{Fig: MultipathPerformance}
\end{figure} 

We now compare the application end-to-end delay of our proposed scheme, with that exhibited by a legacy solution (ARQ scheme) when the feedback delay increases. The obtained results are shown in Fig. \ref{Fig: MultipathDelay}. In this case, we exploit the prediction of the queue at each flow: $\hat{Q}^r_{k,f} = Q^r_{k-d,f} + \sum^{k-1}_{j=k-d}\left( \sum_{i=1}^{n}\left(S_{k,f,i}p_{f,i} - C_{k,f,i}(1-p_{f,i}) \right) \right)$, as was done for the single link scenario. We still use three links and three different paths ($m=n=3$), but now fix the packet loss rate to be $0.2$ over all paths (recall that all flows are equally affected by such erasures). We also assume an arrival rate of $a = \frac{3}{4}$, again for the three flows. The traditional ARQ scheme assumes that the scheduler uses a \emph{Round-Robin} approach to distribute flow transmissions across the paths. It can be seen that the behaviour is much the same as that observed over a single path and, in particular, that as the feedback delay increases, the proposed scheme clearly outperforms ARQ, yielding much lower delays.  

\begin{figure}
	\def \figurewidth {0.8\columnwidth}
	\def \figureheight {3.0cm}
	\pgfplotstableread{figures/data/delay_multipath_3.dat}{\dataset} 

\begin{tikzpicture} 
font = \scriptsize,

\begin{axis}[scale only axis, 
width=\figurewidth,
height=\figureheight,
mark options={solid},
ymin=0,
ymax=250,
xmin=0,
xmax=100,
ylabel={e2e delay (\# slots)},
xlabel={$d$ (\# slots)},
compat=1.3,
legend style={legend pos = north west, draw=none, fill=none , legend columns=1, font = \scriptsize}
]

\addplot [color=black, line width=1pt]
plot [error bars/.cd, y dir = both, y explicit]
table[x index = 0, y index =1, y error plus index=2, y error minus index=2] from \dataset;
\addlegendentry{\emph{Proposal}}

\addplot [dashed,color=black, line width=1pt]
plot [error bars/.cd, y dir = both, y explicit]
table[x index = 0, y index =3, y error plus index=4, y error minus index=4] from \dataset;
\addlegendentry{\emph{ARQ}}

\end{axis}

\end{tikzpicture}
	\caption{Application (e2e) delay Vs. feedback delay ($D_f$) for an ARQ scheme and our proposal. Configuration parameters are arrival rate $a=0.7$, loss rate $p=0.2$ for all paths and $\alpha=1$. The results are averaged after running the experiment during $10^5$ slots and the figure also shows the 95\% confidence intervals}
	\label{Fig: MultipathDelay}
\end{figure}
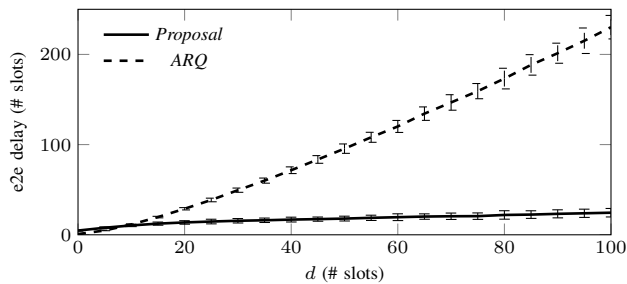

	\section{Conclusions}

In this paper we have proposed a joint coding/scheduling scheme to be used over packet erasure paths. We have posed an optimization problem, which is solved by means of discrete decisions, and the source node can decide: to send a native packet, to transmit a coded packet, or to do nothing. We have shown that this discrete decision method yields an optimum behavior, ensuring in addition system stability. We have assessed the validity of the model, by means of an extensive simulation-based analysis, in which we have considered the impact of having delayed feedback.

Knowing the status of the decoder after some delay has been usually overlooked when studying the performance of coding solutions. For ideal feedback channels it is well known that ARQ yields the best performance. However, under realistic situations, the obtained results have shown that the joint coder/scheduler clearly outperforms legacy solutions. The proposed approach shows the same throughput as the one seen for the ARQ case, while it does not increase the end-to-end delay.

We have also proposed some practical bounds for the corresponding queue lengths, which were afterwards used to analyze the overhead caused by the transmission of unneeded (dummy) packets. The simulation results show that they are indeed rather tight, and that the proposed predictor for the queue occupancy behaves quite accurately. Hence, they can be exploited to take better coding/scheduling decisions in different setups. Last, we have also studied the proposed model over a multi-path communication scenario, where it again outperforms a legacy solution based on ARQ.



%
%

\begin{acronym}

  \acro{TSNC}{Tunable Sparse Network Coding}
  \acro{SNC}{Sparse Network Coding}
  \acro{NC}{Network Coding}
  \acro{RLNC}{Random Linear Network Coding}
  \acro{FER}{Frame Error Rate}
  \acro{RLNC}{Random Linear Network Coding}
  \acro{ITU}{International Telecommunication Union}
  \acro{FEC}{Forward Error Correction}

  \acro{RTT}{Round Trip Time}
  \acro{RTO}{Retransmission TimeOut}
  \acro{WMN}{Wireless Mesh Network}
  \acro{WMNs}{Wireless Mesh Networks}
  \acro{MSS}{Maximum Segment Size}
  \acro{MAC}{Medium Access Control}
  \acro{SDU}{Service Data Unit}
  \acro{PDU}{Protocol Data Unit}
  \acro{DIFS}{Distributed Inter-Frame Space}
  \acro{MTU}{Maximum Transfer Unit}
  \acro{ISP}{Internet Service Provider}
  \acro{RLC}{Random Linear Coding}
  \acro{HMP}{Hidden Markov Process}

  \acro{PER}{Packet Error Rate}

  \acro{PLCP}{Physical Layer Convergence Protocol}

  \acro{ARQ}{Automatic Repeat Request}
  \acro{FEC}{Forward Error Correction}

  \acro{RLSC}{Random Linear Source Coding}

  \acro{MORE}{MAC-independent Opportunistic Routing \& Encoding}
  \acro{CATWOMAN}{Coding Applied To Wireless On Mobile Ad-Hoc Networks}
  \acro{CTCP}{Network Coded TCP}
  
  \acro{cdf}{cumulative distribution function}
  \acro{NACK}{Negative ACKnowledgement}
  \acro{NORM}{NACK-oriented Reliable Multicast}
  \acro{IETF}{Internet Engineering Task Force}
  \acro{COTS}{Commercial of-the-shelf}
  
  \acro{KPI}{Key Performance Indicator}
\end{acronym}


	\section*{Acknowledgements}
	
	This work has been supported by the Spanish Government (Ministerio de Econom\'ia y Competitividad, Fondo Europeo de Desarrollo Regional, FEDER) by means of the project \emph{ADVICE} (TEC2015-71329-C2-1-R). Douglas Leith was supported by Science Foundation Ireland under Grant No. 11/PI/1177 and 13/RC/2077.
	
	\bibliographystyle{IEEEtran}
	\bibliography{LaTeX/biblio}
	
\end{document}